\documentclass[10pt,journal,compsoc]{IEEEtran}

\usepackage{graphicx}
\usepackage{makecell}
\usepackage{algorithmic}
\usepackage[ruled,vlined]{algorithm2e}
\usepackage{multirow}
\usepackage{amsmath}
\usepackage{pifont}
\usepackage{bm}
\usepackage{amsfonts}
\usepackage{epstopdf}
\usepackage{amsthm}
\usepackage{amssymb}
\usepackage{amsmath,amssymb}
\usepackage{color}
\usepackage{subfigure}
\usepackage[justification=centering]{caption}

\ifCLASSOPTIONcompsoc
  \usepackage[nocompress]{cite}
\else
  \usepackage{cite}
\fi

\ifCLASSINFOpdf
\else
\fi

\begin{document}

\title{cuFastTucker: A Compact Stochastic Strategy for Large-scale Sparse Tucker Decomposition on Multi-GPUs}

\author{	
	Zixuan Li
	\IEEEcompsocitemizethanks
	{
		\IEEEcompsocthanksitem
	}
}

\markboth{}%
{Shell \MakeLowercase{\textit{et al.}}: Bare Advanced Demo of IEEEtran.cls for Journals}
\IEEEtitleabstractindextext{%
\begin{abstract}
\underline{H}igh-\underline{O}rder, \underline{H}igh-\underline{D}imension, and \underline{S}parse \underline{T}ensor (HOHDST) data originates from real industrial applications, i.e., social networks, recommender systems, bio-information, and traffic information.
Sparse Tensor Decomposition (STD) can project the HOHDST data into low-rank space.
In this work, a novel method for STD of Kruskal approximating the core tensor and stochastic strategy for approximating the whole gradient is proposed which comprises of the following two parts:
(1) the matrization unfolding order of the Kruskal product for the core tensor follows the multiplication order of the factor matrix and then the proposed theorem can reduce the exponential computational overhead into linear one;
(2) stochastic strategy adopts one-step random sampling set, the volume of which is much smaller than original one, to approximate the whole gradient.
Meanwhile, this method can guarantee the convergence and save the memory overhead.
Due to the compactness of the same order matrix multiplication and parallel access from stochastic strategy, the speed of cuFastTucker can be further reinforced by GPU.
Furthermore, 
a data division and communication strategy of cuFastTucker is proposed for data accommodation on Multi-GPU.
cuFastTucker can achieve the fastest speed and keep the same accuracy and much lower memory overhead than the SOTA algorithms, e.g., P$-$Tucker, Vest, and SGD$\_$Tucker.
The code and partial datasets are publically available on "https://github.com/ZixuanLi-China/FastTucker".
\renewcommand{\raggedright}{\leftskip=0pt \rightskip=0pt plus 0cm}
\raggedright
\end{abstract}

\begin{IEEEkeywords}
GPU CUDA Parallelization;
Kruskal Approximation;
Sparse Tensor Decomposition;
Stochastic Strategy;
Tensor Computation.
\end{IEEEkeywords}}

\maketitle
\IEEEdisplaynontitleabstractindextext
\IEEEpeerreviewmaketitle

\ifCLASSOPTIONcompsoc
\section{Introduction} \label{section1}
\renewcommand{\raggedright}{\leftskip=0pt \rightskip=0pt plus 0cm}
\raggedright
\IEEEPARstart{T}{e}nsors are originated from differential manifold and tensor is used to analyze the change of high dimension space~\cite{ex101skordis2009tensor}.
Due to the amazing representation ability, tensor can capture the relationship between multi-attribute of an entity~\cite{ex101}.
Especially, in Machine Learning (ML) community, which relies on effective statistical learning methodology and plenty of data,
needs powerful data structure to guarantee the abundant information~\cite{ex102, ex110}.
Meanwhile, due to abundant data styles in ML, tensor has drawn wide attention to the emerging ML research communities~\cite{ex102}.

In ML communities, tensor applications can be divided into the following three classes:
(1)In order to capture the general feature of multi-modal data,
muiti-view learning always combines multi-feature into a tensor space~\cite{ex104, ex111, ex112, ex115}.
(2) To project the multi-attribute data into a low complexity and low-rank space,
the learning weight variable always be constituted tensor data, etc,
Tensor Regression~\cite{ex102liu2020smooth, ex103kossaifi2020tensor, ex104yu2016learning},
Support Tensor Machine~\cite{ex105hao2013linear, ex106calvi2019support, ex107biswas2017linear}, and
Deep Convolutional Neural Networks (DCNN) in TensorFlow and Pytorch framework~\cite{ex108abadi2016tensorflow};
(3) Due to the spatiotemporal dynamics and multi-attribute interaction,
the forming data is naturally tensor, e.g,
in Recommendation Systems~\cite{ex108},
Quality of Service (QoS)~\cite{ex109},
Network Flow~\cite{ex105},
Cyber-Physical-Social (CPS)~\cite{ex211}, or
Social Networks~\cite{ex107}.
The scale of tensor data from the fusion process after the multi-modal feature and weight variables of ML methodologies is far below than the natural tensor data.

\underline{H}igh-\underline{O}rder, \underline{H}igh-\underline{D}imension, and \underline{S}parse \underline{T}ensor (HOHDST) is a mathematic model for the data from
Recommendation Systems, QoS, Network Flow, CPS, and Social Networks and high-order and high-dimension mean multi-attribute interaction and multi-entity, respectively \cite{ex113, ex114}.
An $N$-order HOHDST can represent the interaction relationship between $N$ attributes and
in reality, each attribute has millions of entities.
Thus, this property will result in a substantially high-dimension inherence \cite{ex110}.
Unfortunately,
due to data incompleteness,
it is non-trivial to obtain the statistic property of the HOHDST data.

The common used method is finding the low-dimension feature via Sparse Tensor Decomposition (STD) and
this \emph{dimensionality reduction} techniques can represent the original HOHST by low-rank or low-dimension space \cite{ex117, ex113, ex114, ex119}.
Tensor tucker decomposition is one of the most widely used dimensionality reduction methodologies.
Through the $N$-coordinate systems and those systems tangled by a core tensor between each other,
tensor tucker decomposition becomes one of the most used dimensionality reduction methodologies \cite{ex121}.
There are two approaches to find the appropriate core tensor and the $N$ factor matrices:
(1) High Order Orthogonal Iterations (HOOI) should find the $N$ orthogonal coordinate systems and this method needs the Singular Value Decomposition (SVD) for the unfolding tensor.
However, this method relies on frequent Khatri-Rao and Kronecker products for intermediate matrices \cite{ ex120, ex122, ex123, ex214, ex239};
(2) Modern optimization strategy disentangles the tanglement of the core tensor and the $N$ factor matrices and than transfers the non-convex optimization into alternative convex optimization.
The above two methods still involve high-dimension intermediate matrices, and in order to solve these problems,
the main contributions of this work are listed as the following:
\begin{enumerate}
  \item The space overhead of the intermediate coefficient matrices for updating the core tensor is super huge.
        A Kruskal approximation strategy is proposed to divide the core tensor into smaller ones.
        Then, the order of matrix multiplication follows the same matrix multiplication order of the factor matrix.
        Following the proposed computational Theorems \ref{theorem1} and \ref{theorem2}, the computational overhead can be reduced from the exponential overhead into an linear one;
  \item By the one step sampling set, on each training iteration,
        a stochastic strategy is proposed to approximate the whole gradient relying on the whole HOHDST data by partial set.
        This methodology can further reduce the computational overhead and keep the same accuracy;
  \item The fine-grain parallelization inherence gives the allocated CUDA thread block the high parallelization and meanwhile,
        two key steps which are the most time-consuming can be further accelerated by CUDA parallelization (cuFastTucker).
        Because large-scale HOHDST data cannot be accommodated in a single GPU,
        a data division and communication strategy of cuFastTucker is proposed for data accommodation on Multi-GPU.
\end{enumerate}

To our best knowledge,
the proposed model is the first work that it can take advantage of the Kruakal product to approximate the core tensor with linear computational overhead.
In this work,
the related work is presented in Section~\ref{Sectionrela}.
The notations and preliminaries are introduced in Section~\ref{sectionlb}.
The proposed model as well as cuFastTucker are showed in Sections~\ref{section3} and~\ref{section4}, respectively.
Experimental results are illustrated in Section~\ref{section5}.

\section{Related Works} \label{Sectionrela}
ML communities should handle the high-order data and tensor can capture the three or higher-order feature rather than unfolding the high-order data into matrix or vector. 
When the learning data has high-order feature, i.e., Human Recognition Data, Spatiotemporal Dynamics Data, 
Tensor Regression~\cite{ex102liu2020smooth, ex103kossaifi2020tensor, ex104yu2016learning} can project multi-attribute weather data into the forecasting value, and 
Support Tensor Machine~\cite{ex105hao2013linear, ex106calvi2019support, ex107biswas2017linear} can find the discrete classification value from multi-attribute data. 
DCNN plays a key role to learn deep feature from plenty of data, and the tensor decomposition can reduce the parameter complexity~\cite{ex109yin2021towards,ex110panagakis2021tensor}. 
Direct training process for the high-order and high-dimension tensor weight variable will result in over-fitting problem.
To avoid the dimension explosion problem, tensor decomposition can reduce the parameter space overhead and the learning process only involves the factor matrices. 
There are a mass of works try to reduce the parameter complexity. 
However, those methods cannot solve the dimensionality reduction problem in HOHDST data. 
In big data era, it is non-trivial to process the HOHDST data. 
The main problems lie in efficient learning algorithm and the high match process in modern big-data process frameworks, i.e., OpenMP, MPI, CUDA, Hadoop, Spark, and OpenCL, and 
modern hardware, i.e., GPU, CPU, and embedded platforms.

A distributed CANDECOMP/PARAFAC Decomposition (CPD) \cite{ex232} is proposed by Ge, et al., and the CPD is a special Tucker Decomposition for HOHDST. 
In HOHDST data compression community, 
Shaden et al., \cite{ex233} presented a Compressed Sparse Tensors (CSF) structure which can improve the access speed and make data compression for HOHDST.
Ma et al., \cite{ex131} optimized the Tensor-Time-Matrix-chain (TTMc) operation on GPU which is a key part for Tucker Decomposition (TD) and TTMc is
a data intensive task \cite{ex131}. 
A distributed Non-negative Tucker Decomposition (NTD) is proposed by Chakaravarthy et al., \cite{ex130} which needs frequent TTMc operations. 
A parallel strategy of ALS and CD for STD~\cite{ex234, ex237} is presented on OpenMP parallelization. 
A heterogeneous OpenCL parallel version of ALS for STD is proposed on \cite{ex236}. 
The current parallel and distributed works \cite{ex235} mainly focus on divide the whole data into smaller and low-dependence parts and then deploy algorithm rather than fine-grained learning methodologies.


\section{Notations and Preliminaries} \label{sectionlb}
\newtheorem{definition}{Definition}
We denote scalars by regular lowercase or uppercase, vectors by bold lowercase, matrices by bold uppercase, and tensors by bold Euler script letters.
Basic symbols and matrix and tensor operations are presented in Tables \ref{table21} and \ref{table22}, respectively.

\begin{table}[!htbp]
	\setlength{\abovedisplayskip}{0pt}
	\setlength{\belowdisplayskip}{0pt}
	\renewcommand{\arraystretch}{1.5}
	\caption{Table of Tensor Operations.}
	\centering
	\label{table22}
	\tabcolsep1pt
	\begin{tabular}{cc}
		\hline
		\hline
		\makecell[c]{Operations}                         & \makecell[c]{Definition}\\
		\hline
\makecell[c]{The $n$th matricization}
&
\makecell[c]{The $x^{(n)}_{i_{n},j}$ of $\textbf{X}^{(n)}$ where $j=1+\sum_{k=1,n\neq k}^{N}$}\\
\makecell[c]{$\textbf{X}^{(n)}$ of tensor $\mathcal{X}$:}
&
\makecell[c]{$\left[(i_{k}-1)\mathop{\prod}_{m=1,m\neq n}^{k-1}I_{m}\right]$ $=$  $x_{i_{1},i_{2},\cdots,i_{n},\cdots,i_{N}}$;}\\
\makecell[c]{The $n$th column}
&
\makecell[c]{The $x^{(n)}_{k}$ $=$ $\textbf{X}^{(n)}_{i,j}$ of $\textbf{x}^{(n)}$} \\
\makecell[c]{vectorization $\textbf{x}^{(n)}$}
&\\
\makecell[c]{of tensor $\mathcal{X}$:}
&where $k=(j-1)I_{n}+i$; \\
\makecell[c]{$R$ Kruskal Product:}
&
\makecell[c]{$\textbf{B}^{(n)}$ $\in$ $\mathbb{R}^{I_{n}\times R}$, $n \in \{N\}$, $\widehat{\bm{\mathcal{X}}}$ $\in$ $\mathbb{R}^{I_{1}\times\cdots \times I_{N}}$} \\
\makecell[c]{}
&
\makecell[c]{$=\sum_{r=1}^{R} b^{(1)}_{:,r}\circ\cdots\circ b^{(n)}_{:,r}\circ\cdots\circ b^{(N)}_{:,r}$;}\\
\makecell[c]{$n$-Mode}
&
\makecell[c]{$\bm{\mathcal{X}}$ $\in$ $\mathbb{R}^{I_{1}\times\cdots\times I_{N}}$, $\textbf{U}$ $\in$ $\mathbb{R}^{I_{n}\times J_{n}}$ and}\\
\makecell[c]{Tensor-Matrix product:}
&
\makecell[c]{$(\bm{\mathcal{X}}\times_{(n)} \textbf{U})$ $\in$ $\mathbb{R}^{I_{1}\times\cdots \times I_{n-1}\times J_{n}\times\cdots  I_{N}}$}\\
\makecell[c]{}
&
\makecell[c]{$=$ $\sum\limits_{i_{n}=1}^{I_{n}}$ $x_{i_{1}\times\cdots \times i_{n}\times \cdots\times i_{N}}$ $\cdot u_{j_{n},i_{n}}$.}\\
		\hline
		\hline
	\end{tabular}
\end{table}

\begin{table}[!htbp]
	\setlength{\abovedisplayskip}{0pt}
	\setlength{\belowdisplayskip}{0pt}
	\renewcommand{\arraystretch}{1.5}
	\caption{Table of symbols.}
	\centering
	\label{table21}
	\tabcolsep1pt
	\begin{tabular}{cc}
		\hline
		\hline
		\makecell[c]{Symbol}                         & \makecell[c]{Definition}\\
		\hline
		\makecell[c]{$\mathcal{X}$}                  & \makecell[c]{Input $N$th order tensor $\in$ $\mathbb{R}^{I_{1}\times I_{2}\times\cdots \times I_{N}}_{+}$;}\\
		\makecell[c]{$x_{i_{1},i_{2},\cdots,i_{n}}$} & \makecell[c]{$i_{1},i_{2},\cdots,i_{n}$th element of tensor $\mathcal{X}$;}\\
		\makecell[c]{$\mathcal{G}$}                  & \makecell[c]{Core $N$th order tensor $\in$ $\mathbb{R}^{J_{1}\times J_{2}\times\cdots \times J_{N}}$;}\\
		\makecell[c]{$\textbf{X}$}                   & \makecell[c]{Input matrix $\in$ $\mathbb{R}^{I_{1}\times I_{2}}$;}\\
		\makecell[c]{$\{N\}$}                        & \makecell[c]{The ordered set $\{1,2,\cdots,N-1,N\}$;}\\
        \makecell[c]{$\Omega$}                       & \makecell[c]{Index $(i_{1},\cdots,i_{n},\cdots,i_{N})$ of a tensor $\mathcal{X}$;}\\
        \makecell[c]{$\Omega^{(n)}_{M}$}             & \makecell[c]{Index $(i_{n},j)$ of $n$th unfolding matrix $\textbf{X}^{(n)}$;}\\
        \makecell[c]{$(\Omega^{(n)}_{M})_{i}$}       & \makecell[c]{Column index set in $i$th row of $\Omega^{(n)}_{M}$;}\\
        \makecell[c]{$(\Omega^{(n)}_{M})^{j}$}       & \makecell[c]{Row index set in $j$th column of $\Omega^{(n)}_{M}$;}\\
        \makecell[c]{$\Omega^{(n)}_{V}$}             & \makecell[c]{Index $i$ of $n$th unfolding vector Vec$_{n}$($\bm{\mathcal{X}}$);}\\
		\makecell[c]{$\textbf{A}^{(n)}$}             & \makecell[c]{$n$th feature matrix $\in$ $\mathbb{R}^{I_{n}\times J_{n}}$;}\\
		\makecell[c]{$a_{i_{n}, :}^{(n)}$}              & \makecell[c]{$i_{n}$th row vector $\in$ $\mathbb{R}^{K_{n}}$ of $\textbf{A}^{(n)}$;}\\
		\makecell[c]{$b_{:,r}^{(n)}$}                & \makecell[c]{$r$th column vector $\in$ $\mathbb{R}^{K_{n}}$ of $\textbf{B}^{(n)}$;}\\
		\makecell[c]{$a_{i_{n},k_{n}}^{(n)}$}        & \makecell[c]{$k_{n}$th element of feature vector $a_{i_{n}}^{(n)}$;}\\
		\makecell[c]{$\cdot$}                        & \makecell[c]{Element-wise multiplication;}\\
		\makecell[c]{$\circ$}                        & \makecell[c]{Outer production of vectors;}\\
		\makecell[c]{$\odot$}                        & \makecell[c]{Khatri-Rao (columnwise Kronecker) product;}\\
		\makecell[c]{$\times$}                       & \makecell[c]{Matrix product;}\\
		\makecell[c]{$\times_{(n)}$}                 & \makecell[c]{$n$-Mode Tensor-Matrix product;}\\
		\makecell[c]{$\otimes$}                      & \makecell[c]{Kronecker product.}\\
		\hline
		\hline
	\end{tabular}
\end{table}

\subsection{Basic Definitions}
\begin{definition}[The $n$-Rank of a Tensor]
The $n$-Rank of tensor $\bm{\mathcal{X}}$ $\in$ $\mathbb{R}^{I_{1}\times\cdots\times I_{N}}$, is the rank of $n$th matricization $\textbf{X}^{(n)}$, denoted as  $rank_{n}(\bm{\mathcal{X}})$.
	
\end{definition}

\begin{definition}[Tensor Approximation]
For a $N$-order sparse tensor $\bm{\mathcal{X}}$ $\in$ $\mathbb{R}^{I_{1}\times\cdots \times I_{N}}$,  the Tensor Approximation should find a low-rank tensor $\widehat{\bm{\mathcal{X}}}$ $\in$ $\mathbb{R}^{I_{1}\times\cdots \times I_{N}}$ such that
the noisy tensor $\bm{\mathcal{E}}$ $\in$ $\mathbb{R}^{I_{1} \times\cdots\times I_{N}}$ should be small enough, where $\bm{\mathcal{E}}=$$\bm{\mathcal{X}}-\widehat{\bm{\mathcal{X}}}$.
\end{definition}

\begin{definition}[Sparse Tucker Decomposition (STD)]
Given a $N$-order HOHDST $\bm{\mathcal{X}}$ $\in$ $\mathbb{R}^{I_{1}\times\cdots \times I_{N}}$, the goal of STD is
to train a core tensor $\mathcal{G}$ $\in$ $\mathbb{R}^{J_{1}\times\cdots \times J_{N}}$and $N$ factor matrices $\textbf{A}^{(n)}$ $\in$ $\mathbb{R}^{I_{n}\times J_{n}}$,$J_{n} \ll rank_{n}(\mathcal{X})$, $n \in \{N\}$, such that:
\begin{equation}\label{tucker}
  \begin{aligned}
\bm{\mathcal{X}}\approx\widehat{\bm{\mathcal{X}}}=\mathcal{G}&\times_{(1)}\textbf{A}^{(1)}\times_{(2)}\cdots\times_{(n)}\textbf{A}^{(n)}\times_{(n+1)}\cdots\\
&\times_{(N)}\textbf{A}^{(N)}.
\end{aligned}
\end{equation}
The matricized versions of equation (\ref{tucker}) are
\begin{equation}
  \begin{aligned}
\widehat{\bm{X}}^{(n)}=\textbf{A}^{(n)}\textbf{G}^{(n)}\big(\textbf{A}^{(N)}&\otimes\cdots \otimes\textbf{A}^{(n+1)}\otimes\textbf{A}^{(n-1)}\otimes\cdots\\
&\otimes\textbf{A}^{(1)}\big)^{T}
\end{aligned}
\end{equation}
where $\widehat{\bm{X}}^{(n)}$ is $n$th matricization of tensor $\widehat{\bm{\mathcal{\bm{X}}}}$, $\textbf{G}^{(n)}$ is $n$th matricization of tensor $\bm{\mathcal{G}}$ and
\begin{equation}
  \begin{aligned}
\widehat{x}^{(n)} = \big(\textbf{A}^{(N)}&\otimes\cdots \otimes\textbf{A}^{(n+1)}\otimes\textbf{A}^{(n-1)}\otimes\cdots \otimes\textbf{A}^{(1)}\\
&\otimes\textbf{A}^{(n)}\big)g^{(n)},
\end{aligned}
\end{equation}
where $\widehat{x}^{(n)}$ is $n$th vectorization of tensor $\bm{\mathcal{X}}$,
$g^{(n)}$ is $n$th vectorization of tensor $\bm{\mathcal{G}}$.
\end{definition}

The basis optimization problem is organized as \cite{ex225, ex226, ex229, ex230} as:
 \begin{equation}\label{Original}
  \begin{aligned}
  \mathop{\arg\min}_{w \in \mathbb{R}^{R}} f(w)&=\underbrace{L\bigg(w\bigg|y_{i}, x_{i}, w\bigg)}_{Loss\ Function}+ \underbrace{\lambda_{w} R(w)}_{Regularization\ Item}\\
  &=\sum\limits_{i=1}^{N} L_{i}\bigg(w\bigg|y_{i}, x_{i}, w\bigg)+ \lambda_{w}R_{i}(w),
   \end{aligned}
\end{equation}
where $y_{i}$ $\in$ $\mathbb{R}^{1}$, $x_{i}$ $\in$ $\mathbb{R}^{R}$, $i\in \{N\}$, $w \in \mathbb{R}^{R}$.
In the convex optimization community, the literature \cite{ex243} gives the definition of Lipschitz-continuity  with constant $L$ and strong convexity with constant $\mu$.

\begin{definition}[$L$-Lipschitz continuity]
A continuously differentiable function $f(\textbf{x})$ is called $L$-smooth on $\mathbb{R}^{r}$ if the gradient $\nabla f(\textbf{x})$ is $L$-Lipschitz continuous for any $\textbf{x}$, $\textbf{y}$ $\in$ $\mathbb{R}^{r}$, that is
$\|$ $\nabla f(\textbf{x})$ $-$ $\nabla f(\textbf{y})$ $\|_{2}$ $\leq$ $L$ $\|$ $\textbf{x}$ $-$ $\textbf{y}$ $\|_{2}$,
where $\|\bullet\|_{2}$ is $L_{2}$-norm $\|\textbf{x}\|_{2}$ $=$ $(\mathbb{\sum}_{k=1}^{r}x_{k}^{2})^{1/2}$ for a vector $\textbf{x}$.
\end{definition}

\begin{definition}[$\mu$-Convex]
A continuously differentiable function $f(\textbf{x})$ is called strongly-convex on $\mathbb{R}^{r}$ if there exists a constant $\mu$ $>$ $0$ for any $\textbf{x}$, $\textbf{y}$ $\in$ $\mathbb{R}^{r}$, that is
$f(\textbf{x})$ $\geq$ $f(\textbf{y})$ $+$ $\nabla$ $f(\textbf{y})$ $(\textbf{x}-\textbf{y})^{T}$ $+$ $\frac{1}{2}\mu\|\textbf{x}-\textbf{y}\|_{2}^{2}$.
\end{definition}

\begin{definition}[Stochastic Gradient Descent (SGD)]
In large-scale optimization scenarios,
SGD is a common strategy \cite{ex225, ex226, ex229, ex230} and promises to obtain the optimal accuracy via a certain number of training epoches \cite{ex225, ex226, ex229, ex230}.
An $M$ entries set $\Psi$ is randomly selected from the set $\Omega$, and the SGD is presented as:
 \begin{equation}
  \begin{aligned}\label{SGD}
  w&\leftarrow w-\gamma\frac{\partial f_{\Psi(w)}}{\partial w}\\
   &w- \gamma\frac{1}{M}\sum_{i\in\Psi}\frac{\partial \bigg(L_i(w)+\lambda_w R_i(w)\bigg)}{\partial w}.
   \end{aligned}
\end{equation}
\end{definition}

Compared with SGD, the original optimization model needs gradient which should select all the samples $\{x_i|i\in \{N\}\}$ from the dataset $\Omega$.
The optimization function can be packaged in the form of $SGD(M, \lambda, \gamma, w, \frac{\partial f_{\Psi(w)}}{\partial w})$.

\subsection{Optimization of STD}\label{std}
Optimization strategy becomes the most important way to find the optimal feature matrices $\textbf{A}^{(n)}$,$n\in \{N\}$ and core tensor $\bm{\mathcal{G}}$ to suppress the noisy tensor $\bm{\mathcal{E}}$ and elevate the approximation level of $\widehat{\bm{\mathcal{X}}}$, which is presented as:
\begin{equation}\label{tucker_optimization}
\begin{aligned}
&\mathop{\arg\min}_{\textbf{A}^{(n)},n \in \{N\}, \bm{\mathcal{G}}}f\bigg(\bm{X}^{(n)}, \big\{\textbf{A}^{(n)}\big\}, \bm{\mathcal{G}} \bigg)\\
&=\bigg\|\bm{\mathcal{X}}-\widehat{\bm{\mathcal{X}}}\bigg\|_{2}^{2}+\lambda_{\bm{\mathcal{G}}}\|\bm{\mathcal{G}}\|_{2}^{2}+\lambda_{\textbf{A}^{(n)}}\|\textbf{A}^{(n)}\|_{2}^{2}
\end{aligned}
\end{equation}
where $\widehat{\bm{\mathcal{X}}}$ $=$ $\bm{\mathcal{G}}\times_{1}\textbf{A}^{(1)}\times_{2}\cdots\times_{n}\textbf{A}^{(n)}\times_{n+1}\cdots\times_{N}\textbf{A}^{(N)}$ and
$\lambda_{g^{(n)}}$ and $\lambda_{\textbf{A}^{(n)}}$ are the regularization parameters for core tensor and low-rank factor matrices, respectively.

The variables $\bigg\{\{\textbf{A}^{(n)},n \in \{N\}\}, g^{(n)}\bigg\}$ be entangled as the approximated tensor $\widehat{\bm{\mathcal{X}}}$ by tensor-matrix multiplication.
Due to overfitting and non-convex, it is hard to optimize the whole variable $\widehat{\bm{\mathcal{X}}}$.
Alternative optimization is adopted to search the optimal parameters $\bigg\{\{\textbf{A}^{(n)},n \in \{N\}\}, g^{(n)}\bigg\}$ and can obtain appropriate accuracy as:
\begin{equation}\label{low_rank_optimization}
\begin{aligned}
\mathop{\arg\min}_{\textbf{A}^{(n)},n \in\{N\}}
&f\bigg(\textbf{A}^{(n)}\bigg|\bm{X}^{(n)}, \big\{\textbf{A}^{(n)}\big\},\textbf{G}^{(n)}\bigg)\\
&=\bigg\|\bm{X}^{(n)}- \widehat{\bm{X}}^{(n)}\bigg\|_{2}^{2}+\lambda_{\textbf{A}^{(n)}}\|\textbf{A}^{(n)}\|_{2}^{2}
\end{aligned}
\end{equation}
where $\widehat{\bm{X}}^{(n)}$ $=$ $\textbf{A}^{(n)}\textbf{G}^{(n)}\textbf{S}^{(n)^{T}}$ and the coefficient matrix
$\textbf{S}^{(n)}$ $=$ $\textbf{A}^{(N)}\otimes\cdots \otimes\textbf{A}^{(n+1)}\otimes\textbf{A}^{(n-1)}\otimes\cdots \otimes\textbf{A}^{(1)}$ $\in \mathbb{R}^{I_{N}\cdots I_{n+1}I_{n-1}\cdots I_{1}\times J_{N}\cdots J_{n+1}J_{n-1}\cdots J_{1}}$ and $\bm{G}^{(n)}$ $\in$ $\mathbb{R}^{J_{n}\times J_{1}\cdots J_{n-1}J_{n+1}\cdots J_{N}}$.
\begin{equation}\label{core_tensor_optimization}
\begin{aligned}
\mathop{\arg\min}_{g^{(n)}}
&f\bigg(g^{(n)}\bigg|\bm{x}^{(n)}, \big\{\textbf{A}^{(n)}\big\},g^{(n)}\bigg)\\
&=\bigg\|x^{(n)}-\widehat{x}^{(n)}\bigg\|_{2}^{2}+\lambda_{g^{(n)}}\|g^{(n)}\|_{2}^{2}
\end{aligned}
\end{equation}
where $\widehat{x}^{(n)}$ $=$ $\textbf{H}^{(n)}g^{(n)}$ and the coefficient
$\textbf{H}^{(n)}$ $=$ $\textbf{A}^{(N)}$ $\otimes$ $\cdots$ $\otimes$ $\textbf{A}^{(n+1)}$ $\otimes$ $\textbf{A}^{(n-1)}$ $\otimes$ $\cdots$ $\otimes$ $\textbf{A}^{(1)}$ $\otimes$ $\textbf{A}^{(n)}$ $\in$ $\mathbb{R}^{I_{N} \cdots I_{n+1} I_{n-1} \cdots I_{1} I_{n} \times J_{N} \cdots J_{n+1} J_{n-1} \cdots J_{1} J_{n} }$.
The coefficient matrices $\big\{ \{ \textbf{G}^{(n)} , \textbf{S}^{(n)} \} , \{\textbf{H}^{(n)} \} \big\}$ of variables $\{$ $\textbf{A}^{(n)}$ $,$ $g^{(n)}$ $\}$, $n\in\{N\}$, respectively, are memory-consuming.
For the problems (\ref{low_rank_optimization}) and (\ref{core_tensor_optimization}), $L$-lipschitz continuity and $\mu$-convexity can make promise of the convergence.

Besides, it is hard for modern hardware to give consideration to both the computation-orient (GPU) and logic-control-orient (CPU).
To make better use of hardware resource of GPU, algorithm design for high performance computing on GPU should consider high parallelization, low memory overhead, and low conflict probability for the operation of memory Read-and-Write.
Thus, appropriate sampling for coefficient matrix of a optimization strategy should be considered to maintain low overhead and, meanwhile, comparable accuracy.
In the following section, a compact stochastic strategy for STD will be introduced.

\section{A Compact Stochastic Strategy for STD}\label{section3}

The coefficient matrices $\big\{\{\textbf{G}^{(n)}, \textbf{S}^{(n)}\}, \{\textbf{H}^{(n)}\}\big\}$ in the optimization problems (\ref{low_rank_optimization}) and (\ref{core_tensor_optimization}), respectively, are memory-consuming.
Statistic sampling for approximating a full gradient should consider
(1) computational convenience,
(2) convergence and
(3) accuracy.
For the problem 1, SGD for STD chooses the elements from randomly one-step sampling set $\Psi$ which is a subset of index set $\Omega$.
Then, the ingredient of the coefficients $\big\{\{\textbf{G}^{(n)}, \textbf{S}^{(n)}\}, \{\textbf{H}^{(n)}\}\big\}$ just obeys the order from the set $\Psi$ rather than the whole set $\Omega$, and, by this way, the construction overhead from one-step partial set $\Psi$ is much lower than the whole set $\Omega$.

SOTD methods of updating factor matrix and core tensor $\big\{\{\textbf{A}^{(n)}, n\in\{N\}\big\}, \big\{ \mathcal{G} \}\big\}$ are still memory-consuming.
Sections \ref{section31} and \ref{section32} present a novel strategy obeying Theorems \ref{theorem1} and \ref{theorem2} with SGD for the process of updating the factor matrix and core tensor, respectively, and the proposed strategy can further reduce the exponentially increased overhead into linear one.
Hence, due to the compactness of matrix multiplication and parallel access, the proposed model has fine-grained parallelization.
Section \ref{section33} will conclude the overall computational and space overheads.

The gradient of the optimization problem (\ref{tucker_optimization}) should construct a whole coefficient matrix $\textbf{G}^{(n)}$, which is memory-consuming.
The core tensor $\bm{\mathcal{G}}$ can be approximated by $R_{core}$ $\leq$ ${J_{n}, n\in\{N\}}$
Kruskal product of low-rank matrices $\{\textbf{B}^{(n)}\in \mathbb{R}^{J_{n}\times R_{core}}|n\in \{N\}\}$
to form $\widehat{\bm{\mathcal{G}}}$, where
\begin{equation}\label{kruskal_approximation}
\begin{aligned}
\bm{\mathcal{G}} \approx \widehat{\bm{\mathcal{G}}} = \sum_{r_{core}=1}^{R_{core}} b^{(1)}_{:,r_{core}}\circ\cdots\circ b^{(n)}_{:,r_{core}}\circ\cdots\circ b^{(N)}_{:,r_{core}},
\end{aligned}
\end{equation}
and the matricized version of $\widehat{\bm{\mathcal{G}}}$ is
$\widehat{\bm{G}}^{(n)}$ $=$ $\textbf{B}^{(n)}(\textbf{B}^{(N)}\odot\cdots \odot\textbf{B}^{(n+1)}\odot\textbf{B}^{(n-1)}\odot\cdots \odot\textbf{B}^{(1)})^{T}$
$=$ $\sum_{r_{core}=1}^{R_{core}} \textbf{b}^{(n)}_{:,r_{core}}(\textbf{b}^{(N)}_{:,r_{core}}\otimes\cdots\otimes\textbf{b}^{(n+1)}_{:,r_{core}}\otimes
\textbf{b}^{(n-1)}_{:,r_{core}}\otimes\cdots \otimes\textbf{b}^{(1)}_{:,r_{core}})^{T}$,
where $\widehat{\bm{G}}^{(n)}$ is $n$th matricization of tensor $\widehat{\bm{\mathcal{\bm{G}}}}$ and $\widehat{\bm{G}}^{(n)}$ $\in$ $\mathbb{R}^{J_{n}\times J_{1}\cdots J_{n-1}J_{n+1}\cdots J_{N}}$.
In this paper, the problem (\ref{tucker_optimization}) is turned into:
\begin{equation}\label{tucker_cp_optimization}
\begin{aligned}
&\mathop{\arg\min}_{\textbf{A}^{(n)},n \in \{N\}, \widehat{\bm{\mathcal{G}}}}f\bigg(\bm{X}^{(n)}, \big\{\textbf{A}^{(n)}\big\}, \widehat{\bm{\mathcal{G}}} \bigg)\\
&=\bigg\|\bm{\mathcal{X}}-\widehat{\bm{\mathcal{X}}}\bigg\|_{2}^{2}+\lambda_{\widehat{\bm{\mathcal{G}}}}\|\widehat{\bm{\mathcal{G}}}\|_{2}^{2}+\lambda_{\textbf{A}^{(n)}}\|\textbf{A}^{(n)}\|_{2}^{2}.
\end{aligned}
\end{equation}

\subsection{Update Process for Factor Matrix}\label{section31}
Correspondingly, the problem (\ref{low_rank_optimization}) is turned into:
\begin{equation}\label{low_rank_optimization_cp}
\begin{aligned}
\mathop{\arg\min}_{\textbf{A}^{(n)},n \in\{N\}}
&f\bigg(\textbf{A}^{(n)}\bigg|\bm{X}^{(n)}, \big\{\textbf{A}^{(n)}\big\}, \widehat{\textbf{G}}^{(n)}\bigg)\\
&=\bigg\|\bm{X}^{(n)}- \widehat{\bm{X}}^{(n)}\bigg\|_{2}^{2}+\lambda_{\textbf{A}^{(n)}}\|\textbf{A}^{(n)}\|_{2}^{2},
\end{aligned}
\end{equation}
and each feature vector $\textbf{a}^{(n)}_{i_{n}, :}$, $i_{n}\in \{I_{n}\}$, $n \in\{N\}$ shares the same coefficient matrix $\widehat{\textbf{G}}^{(n)}\textbf{S}^{(n)^{T}}$.
With the one-step sampling set $\Psi$, the optimization problem is turned into:
\begin{equation}\label{low_rank_optimization_sgd}
\begin{aligned}
&\mathop{\arg\min}_{\textbf{a}^{(n)}_{i_{n}, :}}
f\bigg(\textbf{a}^{(n)}_{i_{n}, :}\bigg|\bm{X}^{(n)}_{i_n,(\Psi^{(n)}_{M})_{i_n}}, \big\{\textbf{a}^{(n)}_{i_{n}, :}\big\},\widehat{\bm{G}}^{(n)}\bigg)\\
&=\bigg\|\bm{X}^{(n)}_{i_n,(\Psi^{(n)}_{M})_{i_n}}-\widehat{\bm{X}}^{(n)}_{i_n,(\Psi^{(n)}_{M})_{i_n}}\bigg\|_{2}^{2}+\lambda_{\textbf{a}^{(n)}_{i_{n}, :}}\|\textbf{a}^{(n)}_{i_{n}, :}\|_{2}^{2},
\end{aligned}
\end{equation}
where $\widehat{\bm{X}}^{(n)}_{i_n,(\Psi^{(n)}_{M})_{i_n}}$ $=$ $\textbf{a}^{(n)}_{i_{n}, :}\widehat{\bm{G}}^{(n)}\textbf{S}^{(n)T}_{(\Psi^{(n)}_{M})_{i_n}, :}$, and
the $j$th $\bigg($ $j=1+\sum_{k=1,n\neq k}^{N}$ $\left[(i_{k}-1)\mathop{\prod}_{m=1,m\neq n}^{k-1}I_{m}\right]$, $(i_{1},\cdots,i_{n},\cdots,i_{N})$ $\in$ $\Psi$ $\bigg)$ row vector  $\textbf{S}^{(n)}_{i_n, :}$ $=$ $\textbf{a}^{(N)}_{i_N, :}\otimes\cdots \otimes\textbf{a}^{(n+1)}_{i_{n+1}, :}\otimes\textbf{a}^{(n-1)}_{i_{n-1}, :}\otimes\cdots \otimes\textbf{a}^{(1)}_{i_1, :}$ $\in \mathbb{R}^{J_{N}\cdots J_{n+1}J_{n-1}\cdots J_{1}}$ and the approximated gradient from SGD is obtained as:
\begin{equation}\label{Gradient_low_rank}
\begin{aligned}
&\frac{\partial f\bigg(\textbf{a}^{(n)}_{i_{n}, :}\bigg|\bm{X}^{(n)}_{i_n,(\Psi^{(n)}_{M})_{i_n}}, \big\{\textbf{a}^{(n)}_{i_{n}, :}\big\},\widehat{\bm{G}}^{(n)}\bigg)}{\partial \textbf{a}^{(n)}_{i_{n}, :}}\\
=&-\underbrace{\bm{X}^{(n)}_{i_n,(\Psi^{(n)}_{M})_{i_n}}\textbf{D}^{(n)T}_{(\Psi^{(n)}_{M})_{i_n}, :}}_{(1): FacMatPart_1 \in \mathbb{R}^{J_{n}}}+
\underbrace{\lambda_{\textbf{a}^{(n)}_{i_{n}, :}}\textbf{a}^{(n)}_{i_{n}, :}}_{(2): FacMatPart_2 \in \mathbb{R}^{J_{n}}}\\
&+\underbrace{\overbrace{\textbf{a}^{(n)}_{i_{n}, :}\textbf{D}^{(n)}_{(\Psi^{(n)}_{M})_{i_n}, :}}^{InterMX: \widehat{\bm{X}}^{(n)}_{i_n,(\Psi^{(n)}_{M})_{i_n}}}\textbf{D}^{(n)T}_{(\Psi^{(n)}_{M})_{i_n}, :}}_{(3): FacMatPart_3 \in \mathbb{R}^{J_{n}}},
\end{aligned}
\end{equation}
where $\textbf{D}^{(n)}_{(\Psi^{(n)}_{M})_{i_n}, :}$ $=$ $\widehat{\bm{G}}^{(n)}\textbf{S}^{(n)T}_{(\Psi^{(n)}_{M})_{i_n}, :}$ $\in$ $\mathbb{R}^{J_{n}\times |(\Psi^{(n)}_{M})_{i_n}|}$.

\newtheorem{theorem}{Theorem}
\begin{theorem}\label{theorem1}

There are two row vectors
$x= x^{(N)}\otimes x^{(N-1)}\otimes\cdots\otimes x^{(n)}\otimes\cdots \otimes x^{(2)} \otimes x^{(1)}$ and
$y= y^{(N)}\otimes y^{(N-1)}\otimes\cdots\otimes y^{(n)}\otimes\cdots \otimes y^{(2)} \otimes y^{(1)}$, $n\in\{N\}$, where
$x^{(n)}$ $\in$ $\mathbb{R}^{I_{n}}$ and
$y^{(n)}$ $\in$ $\mathbb{R}^{I_{n}}$,
$x$ $\in$ $\mathbb{R}^{I_{1}\cdots I_{n}\cdots I_{N}}$,
$y$ $\in$ $\mathbb{R}^{I_{1}\cdots I_{n}\cdots I_{N}}$, $n\in\{N\}$.
The vector-vector multiplication
$xy^{T}$ $=$
$\big(x^{(N)}\otimes x^{(N-1)}\otimes\cdots\otimes x^{(n)}\otimes\cdots \otimes x^{(2)} \otimes x^{(1)}\big)$
$\big(y^{(N)}\otimes y^{(N-1)}\otimes\cdots\otimes y^{(n)}\otimes\cdots \otimes y^{(2)} \otimes y^{(1)}\big)^{T}$ can be transformed into
$xy^{T}$ $=$
$\big(x^{(N)}y^{(N)^{T}}\big)$ $\big(x^{(N-1)}y^{(N-1)^{T}}\big)$ $\cdots$ $\big(x^{(n)}y^{(n)^{T}}\big)$ $\cdots$  $\big(x^{(2)}y^{(2)^{T}}\big)$ $\big(x^{(1)}y^{(1)^{T}}\big)$.
\end{theorem}
\begin{proof}
The index $\big(i/j\big)$ corresponds to a solely index $\biggl((i_{1}^{x},\cdots,i_{n}^{x},\cdots,i_{N}^{x})$ $\biggl/$ $(i_{1}^{y},\cdots,i_{n}^{y},\cdots,i_{N}^{y})\biggl)$,respectively, where
$\biggl(i=1+\sum_{k=1}^{N}\left[(i_{k}^{x}-1)\mathop{\prod}_{m=1,m\neq 1}^{k-1}I_{m}\right]$ $\biggl/$
$j=1+\sum_{k=1}^{N}\left[(i_{k}^{y}-1)\mathop{\prod}_{m=1, m\neq 1}^{k-1}I_{m}\right]\biggl)$ and
$\biggl(x_{i}=\prod_{n=1}^{N}x^{(n)}_{i_{n}^{x}}$ $\biggl/$
$y_{j}=\prod_{n=1}^{N}y^{(n)}_{j_{n}^{y}}\biggl)$.
$xy^{T}$ $=$ $\sum\limits_{i=j=1}^{I_{1}\cdots I_{n}\cdots I_{N}}x_{i}y_{j}$ $=$
$\sum\limits_{i_{1}^{x}=j_{1}^{y}=1}^{I_{1}} \cdots \sum\limits_{i_{n}^{x}=j_{n}^{y}=1}^{I_{n}} \cdots \sum\limits_{i_{N}^{x}=j_{N}^{y}=1}^{I_{N}}$
$\bigg(\prod_{n=1}^{N}x^{(n)}_{i_{n}^{x}}\prod_{n=1}^{N}y^{(n)}_{j_{n}^{y}}\bigg)$
$=$
$\big(x^{(N)}y^{(N)^{T}}\big)$ $\cdot$ $\big(x^{(N-1)}y^{(N-1)^{T}}\big)$ $\cdot$ $\cdots$ $\cdot$ $\big(x^{(n)}y^{(n)^{T}}\big)$ $\cdot$ $\cdots$ $\cdot$ $\big(x^{(2)}y^{(2)^{T}}\big)$ $\cdot$ $\big(x^{(1)}y^{(1)^{T}}\big)$.
\end{proof}

According to Theorem \ref{theorem1},
$\bm{GS}^{(n)}_{:,j}=\widehat{\bm{G}}^{(n)}\textbf{S}^{(n)T}_{(\Psi^{(n)}_{M})_{i_n}, j}$
$=$
$\textbf{B}^{(n)}\big(\textbf{B}^{(N)}\odot\cdots \odot\textbf{B}^{(n+1)}\odot\textbf{B}^{(n-1)}\odot\cdots \odot\textbf{B}^{(1)}\big)^{T}$
$\big($ $\textbf{a}^{(N)}_{i_N, :}\otimes\cdots \otimes\textbf{a}^{(n+1)}_{i_{n+1}, :}\otimes\textbf{a}^{(n-1)}_{i_{n-1}, :}\otimes\cdots \otimes\textbf{a}^{(1)}_{i_1, :}$ $\big)^{T}$
$=$
$\sum_{r_{core}=1}^{R_{core}} \textbf{b}^{(n)}_{:,r_{core}}\big(\textbf{b}^{(N)}_{:,r_{core}}\otimes\cdots\otimes\textbf{b}^{(n+1)}_{:,r_{core}}\otimes
\textbf{b}^{(n-1)}_{:,r_{core}}\otimes\cdots \otimes\textbf{b}^{(1)}_{:,r_{core}}\big)^{T}$
$\big($ $\textbf{a}^{(N)}_{i_N, :}\otimes\cdots \otimes\textbf{a}^{(n+1)}_{i_{n+1}, :}\otimes\textbf{a}^{(n-1)}_{i_{n-1}, :}\otimes\cdots \otimes\textbf{a}^{(1)}_{i_1, :}$ $\big)^{T}$
$=$
$\sum_{r_{core}=1}^{R_{core}}$ $\textbf{b}^{(n)}_{:,r_{core}}$
$\big($
$(\underbrace{\textbf{b}^{(N)T}_{:,r_{core}}\textbf{a}^{(N)T}_{i_N, :}}_{c^{(N)}_{r_{core}}})$ $\cdot$ $\cdots$ $\cdot$
$(\underbrace{\textbf{b}^{(n+1)T}_{:,r_{core}}\textbf{a}^{(n+1)T}_{i_{n+1}, :}}_{c^{(n+1)}_{r_{core}}})$
$(\underbrace{\textbf{b}^{(n-1)T}_{:,r_{core}}\textbf{a}^{(n-1)T}_{i_{n-1}, :}}_{c^{(n-1)}_{r_{core}}})$ $\cdot$ $\cdots$ $\cdot$
$(\underbrace{\textbf{b}^{(1)T}_{:,r_{core}}\textbf{a}^{(1)T}_{i_1, :}}_{c^{(1)}_{r_{core}}})$
$\big)$, where $j=1+\sum_{k=1,n\neq k}^{N}$ $\left[(i_{k}-1)\mathop{\prod}_{m=1,m\neq n}^{k-1}I_{m}\right]$, $(i_{1},\cdots,i_{n},\cdots,i_{N})$ $\in$ $\Psi$.

\subsection{Update Process for Core Tensor}\label{section32}
Meanwhile, the problem (\ref{core_tensor_optimization}) is turned into:
\begin{equation}\label{core_tensor_optimization_cp}
\begin{aligned}
\mathop{\arg\min}_{\widehat{g}^{(n)}}
&f\bigg(\widehat{g}^{(n)}\bigg|x^{(n)}, \big\{\textbf{A}^{(n)}\big\}, \widehat{g}^{(n)}\bigg)\\
&=\bigg\|\bm{X}^{(n)}-\widehat{\bm{X}}^{(n)}\bigg\|_{2}^{2}+\lambda_{\widehat{g}^{(n)}}\|\widehat{g}^{(n)}\|_{2}^{2}\\
&=\sum_{i_{n}=1}^{I_{n}}\bigg\|\bm{X}^{(n)}_{i_n, :}-\widehat{\bm{X}}^{(n)}_{i_n, :}\bigg\|_{2}^{2}+\lambda_{\widehat{g}^{(n)}}\|\widehat{g}^{(n)}\|_{2}^{2}.
\end{aligned}
\end{equation}
where $\widehat{\bm{X}}^{(n)}_{i_n, :}$
$\mathop{=}$
$\textbf{a}^{(n)}_{i_{n}, :}\widehat{\bm{G}}^{(n)}\textbf{S}^{(n)T}$
$\mathop{=}$
$\textbf{a}^{(n)}_{i_{n}, :}$ $\sum_{r_{core}=1}^{R_{core}}$ $\textbf{b}^{(n)}_{:,r_{core}}$ $\big(\textbf{b}^{(N)}_{:,r_{core}}\otimes\cdots\otimes\textbf{b}^{(n+1)}_{:,r_{core}}\otimes
\textbf{b}^{(n-1)}_{:,r_{core}}\otimes\cdots \otimes\textbf{b}^{(1)}_{:,r_{core}}\big)^{T}$ $\textbf{S}^{(n)T}$
$\mathop{=}$
$\sum_{r_{core}=1}^{R_{core}}$ $\underbrace{\textbf{a}^{(n)}_{i_{n}, :} \textbf{b}^{(n)}_{:,r_{core}}}_{\textbf{b}^{(n)T}_{:,r_{core}}\textbf{a}^{(n)T}_{i_{n}, :}}$ $\big(\textbf{b}^{(N)}_{:,r_{core}}\otimes\cdots\otimes\textbf{b}^{(n+1)}_{:,r_{core}}\otimes
\textbf{b}^{(n-1)}_{:,r_{core}}\otimes\cdots \otimes\textbf{b}^{(1)}_{:,r_{core}}\big)^{T}$ $\textbf{S}^{(n)T}$
$\mathop{=}$
$\sum_{r_{core}=1}^{R_{core}}$ $\textbf{b}^{(n)T}_{:,r_{core}}$ $\underbrace{\textbf{a}^{(n)T}_{i_{n}, :} (\textbf{b}^{(N)}_{:,r_{core}}\otimes\cdots\otimes\textbf{b}^{(n+1)}_{:,r_{core}}\otimes
\textbf{b}^{(n-1)}_{:,r_{core}}\otimes\cdots \otimes\textbf{b}^{(1)}_{:,r_{core}})^{T} \textbf{S}^{(n)T} }_{\textbf{Q}^{r_{core}} \in R^{J_n \times I_{N}\cdots I_{n+1}I_{n-1}\cdots I_{1}}}$
, according to Equ. (\ref{kruskal_approximation}).
\begin{theorem}\label{theorem2}
Assume a row vector
$x= x^{(N)}\otimes x^{(N-1)}\otimes\cdots\otimes x^{(n)}\otimes\cdots \otimes x^{(2)} \otimes x^{(1)}$ and a matrix
$\textbf{Y}= \textbf{Y}^{(N)}\otimes \textbf{Y}^{(N-1)}\otimes\cdots\otimes \textbf{Y}^{(n)}\otimes\cdots \otimes \textbf{Y}^{(2)} \otimes \textbf{Y}^{(1)}$, $n\in\{N\}$, where
$x^{(n)}$ $\in$ $\mathbb{R}^{I_{n}}$ and
$\textbf{Y}^{(n)}$ $\in$ $\mathbb{R}^{J_{n}\times I_{n}}$,
$x$ $\in$ $\mathbb{R}^{I_{1}\cdots I_{n}\cdots I_{N}}$,
$\textbf{Y}$ $\in$ $\mathbb{R}^{J_{1}\cdots J_{n}\cdots J_{N}\times I_{1}\cdots I_{n}\cdots I_{N}}$, $n\in\{N\}$.
The vector-matrix multiplication
$x\textbf{Y}^{T}$ $=$
$\big(x^{(N)}\otimes x^{(N-1)}\otimes\cdots\otimes x^{(n)}\otimes\cdots \otimes x^{(2)} \otimes x^{(1)}\big)$
$\big(\textbf{Y}^{(N)}\otimes \textbf{Y}^{(N-1)}\otimes\cdots\otimes \textbf{Y}^{(n)}\otimes\cdots \otimes \textbf{Y}^{(2)} \otimes \textbf{Y}^{(1)}\big)^{T}$ can be transformed into
$x\textbf{Y}^{T}$ $=$
$\big(x^{(N)}\textbf{Y}^{(N)^{T}}\big)$ $\otimes$ $\big(x^{(N-1)}\textbf{Y}^{(N-1)^{T}}\big)$ $\otimes$ $\cdots$ $\otimes$ $\big(x^{(n)}\textbf{Y}^{(n)^{T}}\big)$ $\otimes$ $\cdots$ $\otimes$  $\big(x^{(2)}\textbf{Y}^{(2)^{T}}\big)$ $\otimes$ $\big(x^{(1)}\textbf{Y}^{(1)^{T}}\big)$.
\end{theorem}
\begin{proof}
The $i$th element of $x\textbf{Y}^{T}$ is $xy^{T}_{i,:}$, where $i=1+\sum_{k=1}^{N}\left[(i_{k}-1)\mathop{\prod}_{m=1,m\neq 1}^{k-1}I_{m}\right]$, $i_{n}\in \{I_{n}\}, n\in\{N\}$, and $y_{i, :}=
y_{i_N ,:}^{(N)}\otimes y_{i_{N-1},:}^{(N-1)}\otimes\cdots\otimes y_{i_n, :}^{(n)}\otimes\cdots \otimes y_{i_2,:}^{(2)} \otimes y_{i_1,:}^{(1)}$.
According to Theorem \ref{theorem1},
$xy_{i, :}^{T}$
$=$
$\big(x^{(N)}y_{i_{N},:}^{(N)^{T}}\big)$ $\cdot$ $\big(x^{(N-1)}y_{i_{N-1},:}^{(N-1)^{T}}\big)$ $\cdot$ $\cdots$ $\cdot$ $\big(x^{(n)}y_{i_n,:}^{(n)^{T}}\big)$ $\cdot$ $\cdots$ $\cdot$ $\big(x^{(2)}y_{i_2,:}^{(2)^{T}}\big)$ $\cdot$ $\big(x^{(1)}y_{i_1,:}^{(1)^{T}}\big)$
$=$
$\big(x^{(N)}y_{i_{N},:}^{(N)^{T}}\big)$ $\otimes$ $\big(x^{(N-1)}y_{i_{N-1},:}^{(N-1)^{T}}\big)$ $\otimes$ $\cdots$ $\otimes$ $\big(x^{(n)}y_{i_n,:}^{(n)^{T}}\big)$ $\otimes$ $\cdots$ $\otimes$  $\big(x^{(2)}y_{i_2,:}^{(2)^{T}}\big)$ $\otimes$ $\big(x^{(1)}y_{i_1,:}^{(1)^{T}}\big)$.
\end{proof}
According to Theorem \ref{theorem2} and the definition of $\textbf{S}^{(n)}$ in Equ. \ref{low_rank_optimization},
$(\textbf{b}^{(N)}_{:,r_{core}}\otimes\cdots\otimes\textbf{b}^{(n+1)}_{:,r_{core}}\otimes
\textbf{b}^{(n-1)}_{:,r_{core}}\otimes\cdots \otimes\textbf{b}^{(1)}_{:,r_{core}})^{T}$ $\textbf{S}^{(n)T}$
$=$
$(\textbf{b}^{(N)T}_{:,r_{core}}\textbf{A}^{(N)T}) \otimes \cdots \otimes(\textbf{b}^{(n+1)T}_{:,r_{core}}\textbf{A}^{(n+1)T})\otimes
(\textbf{b}^{(n-1)T}_{:,r_{core}}\textbf{A}^{(n-1)T})\otimes\cdots \otimes(\textbf{b}^{(1)T}_{:,r_{core}}\textbf{A}^{(1)T})$.
With $\textbf{Q}^{r_{core}}$
$=$
$(\textbf{b}^{(N)T}_{:,r_{core}}\textbf{A}^{(N)T})$ $\otimes$ $\cdots$ $\otimes$
$(\textbf{b}^{(n+1)T}_{:,r_{core}}\textbf{A}^{(n+1)T})$ $\otimes$
$(\textbf{b}^{(n-1)T}_{:,r_{core}}\textbf{A}^{(n-1)T})$ $\otimes$ $\cdots$ $\otimes$
$(\textbf{b}^{(1)T}_{:,r_{core}}\textbf{A}^{(1)T})$ and
$\widehat{\bm{X}}^{(n)}_{i_n, :}$ $=$ $\sum_{r_{core}=1}^{R_{core}}$ $\textbf{b}^{(n)T}_{:,r_{core}}$ $\textbf{Q}^{r_{core}}$,
the problem (\ref{core_tensor_optimization_cp}) is transformed into:
\begin{equation}\label{core_tensor_optimization_cp1}
\begin{aligned}
\mathop{\arg\min}_{\textbf{b}^{(n)T}_{:,r_{core}}}
&f\bigg(\textbf{b}^{(n)T}_{:,r_{core}}\bigg|x^{(n)}, \big\{\textbf{A}^{(n)}\big\}, \widehat{g}^{(n)}\bigg)\\
=&\sum_{i_{n}=1}^{I_{n}}\bigg\|\bm{X}^{(n)}_{i_n, :}-\sum_{r_{core}=1}^{R_{core}}\textbf{b}^{(n)T}_{:,r_{core}}\textbf{Q}^{r_{core}}\bigg\|_{2}^{2}\\
 &+\lambda_{\textbf{b}^{(n)}_{:,r_{core}}}\|\textbf{b}^{(n)T}_{:,r_{core}}\|_{2}^{2},
\end{aligned}
\end{equation}
where $\textbf{b}^{(n)T}_{:,r_{core}}, n\in\{N\}, r_{core}\in\{R_{core}\}$.
However, the computational and space overheads to construct the gradient are still high.
With the one-step sampling set $\Psi$ and fixed $r_{core}=r$, $r\in\{R_{core}\}$, the problem (\ref{core_tensor_optimization_cp1}) is turned into:
\begin{equation}\label{low_rank_optimization_core_tensor_sgd}
\begin{aligned}
&\mathop{\arg\min}_{\textbf{b}^{(n)T}_{:,r}}
f\bigg(\textbf{b}^{(n)T}_{:,r}\bigg|x^{(n)}_{\Psi^{(n)}_{V}}, \big\{\textbf{A}^{(n)}\big\}, \widehat{g}^{(n)}\bigg)\\
=&\sum_{i_{n}=1}^{I_{n}}\bigg\|\bm{X}^{(n)}_{i_n,(\Psi^{(n)}_{M})_{i_n}}-\textbf{b}^{(n)T}_{:,r}\textbf{Q}^{r}_{:,(\Psi^{(n)}_{M})_{i_n}}\\
&-\sum_{r_{core}=1, r_{core}\neq r}^{R_{core}}\textbf{b}^{(n)T}_{:,r_{core}}\textbf{Q}^{r_{core}}_{:, (\Psi^{(n)}_{M})_{i_n}}\bigg\|_{2}^{2}\\
&+\lambda_{\textbf{b}^{(n)}_{:,r}}\|\textbf{b}^{(n)T}_{:,r}\|_{2}^{2},
\end{aligned}
\end{equation}
and the approximated gradient from SGD is obtained as:
\begin{equation}\label{Gradient_low_rank_core}
\begin{aligned}
&\frac{\partial f\bigg(\textbf{b}^{(n)T}_{:,r}\bigg|x^{(n)}_{\Psi^{(n)}_{V}}, \big\{\textbf{A}^{(n)}\big\}, \widehat{g}^{(n)}\bigg)}{\partial \textbf{b}^{(n)T}_{:,r}}\\
=&-\sum_{i_{n}=1}^{I_{n}}\underbrace{\bm{X}^{(n)}_{i_n,(\Psi^{(n)}_{M})_{i_n}}\textbf{Q}^{(n), rT}_{:,(\Psi^{(n)}_{M})_{i_n}}}_{(1): CoreTensorPart_1}
+\underbrace{\lambda_{\textbf{b}^{(n)}_{:,r}}\textbf{b}^{(n)T}_{:,r}}_{(2): CoreTensorPart_2}\\
&+\sum_{i_{n}=1}^{I_{n}}\bigg\{\underbrace{\overbrace{\textbf{b}^{(n)T}_{:,r}\textbf{Q}^{(n), r}_{:, (\Psi^{(n)}_{M})_{i_n}}}^{InterMX_r}\textbf{Q}^{(n), rT}_{:, (\Psi^{(n)}_{M})_{i_n}}}_{(3): CoreTensorPart_3}\\
&+\underbrace{\sum_{r_{core}=1\neq r}^{R_{core}}\bigg(\overbrace{\textbf{b}^{(n)T}_{:,r_{core}}\textbf{Q}^{(n), r_{core}}_{:, (\Psi^{(n)}_{M})_{i_n}}}^{InterMX_{r_{core}}}\textbf{Q}^{(n), rT}_{:, (\Psi^{(n)}_{M})_{i_n}}}_{(4): CoreTensorPart_4}\bigg)\bigg\},\\
\end{aligned}
\end{equation}
where
$\textbf{Q}^{(n), r}_{:, j}$ $\in$ $\mathbb{R}^{J_{n}}$, and
$\textbf{Q}^{(n), r}_{:, j}$
$=$
$\textbf{a}^{(n)T}_{i_{n}, :}$
$\big(\textbf{b}^{(N)}_{:,r}\otimes\cdots\otimes\textbf{b}^{(n+1)}_{:,r}\otimes
\textbf{b}^{(n-1)}_{:,r}\otimes\cdots \otimes\textbf{b}^{(1)}_{:,r}\big)^{T}$
$\big($ $\textbf{a}^{(N)}_{i_N, :}\otimes\cdots \otimes\textbf{a}^{(n+1)}_{i_{n+1}, :}\otimes\textbf{a}^{(n-1)}_{i_{n-1}, :}\otimes\cdots \otimes\textbf{a}^{(1)}_{i_1, :}$ $\big)^{T}$
$=$
$\textbf{a}^{(n)T}_{i_{n}, :}$
$\big($
$(\underbrace{\textbf{b}^{(N)T}_{:,r}\textbf{a}^{(N)T}_{i_N, :}}_{c_r^{(N)}})$ $\cdot$ $\cdots$ $\cdot$
$(\underbrace{\textbf{b}^{(n+1)T}_{:,r}\textbf{a}^{(n+1)T}_{i_{n+1}, :}}_{c_{r}^{(n+1)}})$
$(\underbrace{\textbf{b}^{(n-1)T}_{:,r}\textbf{a}^{(n-1)T}_{i_{n-1}, :}}_{c_{r}^{(n-1)}})$ $\cdot$ $\cdots$ $\cdot$
$(\underbrace{\textbf{b}^{(1)T}_{:,r}\textbf{a}^{(1)T}_{i_1, :}}_{c_r^{(1)}})$
$\big)$,
where $j=1+\sum_{k=1,n\neq k}^{N}$ $\left[(i_{k}-1)\mathop{\prod}_{m=1,m\neq n}^{k-1}I_{m}\right]$, $(i_{1},\cdots,i_{n},\cdots,i_{N})$ $\in$ $\Psi$, $n\in\{N\}$, $r\in\{R_{core}\}$.

\begin{table}[htbp]
	\setlength{\abovedisplayskip}{0pt}
	\setlength{\belowdisplayskip}{0pt}
	\renewcommand{\arraystretch}{1.5}
	\caption{Table of Computational Complexity.}
	\centering
	\label{Computational_Complexity}
	\tabcolsep1pt
	\begin{tabular}{cc}
		\hline
		\hline
		\makecell[c]{Updating Factor Matrices}                         & \makecell[c]{Computational Complexity}\\
		\makecell[c]{$i_{n}$ $\in$ $\{I_{n}\}$, $n$ $\in$ $\{N\}$}     & \\
		\hline
        \makecell[c]{$\textbf{D}^{(n)}_{(\Psi^{(n)}_{M})_{i_n}, j}$}
         &
        \makecell[c]{$O\big(DCC^{(n)}\big)$;}\\
        \makecell[c]{$\textbf{D}^{(n)}_{(\Psi^{(n)}_{M})_{i_n}, :}$}
        &
        \makecell[c]{$O\big(|(\Psi^{(n)}_{M})_{i_n}|DCC^{(n)}\big)$;}\\
        \makecell[c]{$\bm{X}^{(n)}_{i_n,(\Psi^{(n)}_{M})_{i_n}}\textbf{D}^{(n)}_{(\Psi^{(n)}_{M})_{i_n}, :}$}
        &
        \makecell[c]{$O\big(J_n|(\Psi^{(n)}_{M})_{i_n}|\big)$;}\\
        \makecell[c]{$\textbf{E}^{(n)}_{i_n}=\textbf{D}^{(n)}_{(\Psi^{(n)}_{M})_{i_n}, :}\textbf{D}^{(n)T}_{(\Psi^{(n)}_{M})_{i_n}, :}$}
        &
        \makecell[c]{$O\big(J_n^{2}|(\Psi^{(n)}_{M})_{i_n}|\big)$;}\\
        \makecell[c]{$\textbf{a}^{(n)}_{i_{n}, :}\textbf{E}^{(n)}_{i_n}$}
        &
        \makecell[c]{$O\big(J_n^{2}\big)$;}\\
		\hline
        \makecell[c]{Total}
        &
        \makecell[c]{$O\bigg(\sum\limits_{n=1}^{N}\big(|\Psi^{(n)}_{V}|(J_n+J^2_n$}\\
        \makecell[c]{for all $n\in\{N\}$}
        &
        \makecell[c]{$+DCC^{(n)})+J^2_n\big)\bigg)$.}\\
		\hline
		\makecell[c]{Updating Core Tensor}                         & \\
		\makecell[c]{Fixing a $n$ $\in$ $\{N\}$, $i_{n}$ $\in$ $\{I_{n}\}$} & \\
		\makecell[c]{and $r$ $\in$ $\{R_{core}\}$} & \\
		\hline
        \makecell[c]{$\textbf{Q}^{(n), r}_{:, j}$}
         &
        \makecell[c]{$O\big(QCC^{(n)}\big)$;}\\
        \makecell[c]{$\textbf{Q}^{(n), r}_{:,(\Psi^{(n)}_{M})_{i_n}}$}
         &
        \makecell[c]{$O\big(|(\Psi^{(n)}_{M})_{i_n}|QCC^{(n)}\big)$;}\\
        \makecell[c]{Part (1)}
         &
        \makecell[c]{$O\big(J_{n}|(\Psi^{(n)}_{M})_{i_n}|\big)$;}\\
        \makecell[c]{Part (2)}
         &
        \makecell[c]{$O\big(J_{n}\big)$;}\\
        \makecell[c]{Part (3)}
         &
        \makecell[c]{$O\big(J_n^{2}|(\Psi^{(n)}_{M})_{i_n}|+J_n^{2}\big)$;}\\
        \makecell[c]{Part (4)}
         &
        \makecell[c]{$O\bigg((R_{core}-1)\big(J_n^{2}|(\Psi^{(n)}_{M})_{i_n}|+J_n^{2}\big)\bigg)$;}\\
		\hline
        \makecell[c]{Total}
        &
        \makecell[c]{$O\bigg(\sum\limits_{n=1}^{N}R_{core}^2|\Psi^{(n)}_{V}|\big(J_n+J^2_n$}\\
        \makecell[c]{for all $n\in\{N\}$ and }
        &
        \makecell[c]{$+QDCC^{(n)}\big)+\sum\limits_{n=1}^{N}R_{core}I_nJ_{n}\bigg)$.}\\
        \makecell[c]{$r_{core}$ $\in$ $\{R_{core}\}$}
        &
        \makecell[c]{}\\
        \hline
		\hline
	\end{tabular}
\end{table}

\begin{algorithm}[hptb]
	\caption{Algorithm for Updating Factor Matrices and Core Tensor.}
	\label{sgd_sftd_a}
	\vspace{.1cm}
	$\textbf{Input}$: Sparse tensor $\mathcal{X}$, learning rate $\gamma_{\textbf{A}}$,
	regularization parameter $\lambda_{\textbf{A}}$.
	Initializing $\textbf{A}^{(n)}$, $\textbf{B}^{(n)}$, sampling set $\Psi$ and $M=1$ for factor matrices, $M=|\Psi|$ for core tensor, and each step selecting an index $($ $i_{1}$, $\cdots$, $i_{n}$, $\cdots$, $i_{N}$ $)$ from $\Psi$, where $n$ $\in$ $\{N\}$.\\
	$\textbf{Output}$: $\textbf{A}^{(n)}$, $\textbf{B}^{(n)}$, $n$ $\in$ $\{N\}$.\\
	\begin{algorithmic}[1]
	\FOR{$n$ from $1$ to $N$}
        \STATE Set all $c_r^{(n)}\leftarrow 0$, $\bm{GS}^{(n)}_{:}\leftarrow 0$, $n\in\{N\}$, $r\in\{R_{core}\}$;
        \STATE $j\leftarrow 1+\sum_{k=1,n\neq k}^{N}$ $\left[(i_{k}-1)\mathop{\prod}_{m=1,m\neq n}^{k-1}I_{m}\right]$;
		\FOR{$r_{core}$ from $1$ to $R_{core}$}
		  \FOR{$n_{0}$ from $1$ to $N$, $n_{0}$ $\neq$ $n$}
               \STATE $c_{r_{core}}^{(n_{0})}\leftarrow \textbf{b}^{(n_{0})T}_{:, r_{core}}\textbf{a}^{(n_{0})T}_{i_{n_{0}}, :}$; \%This step can be accelerated by CUDA Warp Shuffle and Memory Coalescing.
               \STATE $\textbf{b}^{(n)}_{:,r_{core}}\leftarrow c_{r_{core}}^{(n_{0})}\textbf{b}^{(n)}_{:,r_{core}}$;
		  \ENDFOR
       \STATE $\bm{GS}^{(n)}_{:}\leftarrow\bm{GS}^{(n)}_{:}+\textbf{b}^{(n)}_{:,r_{core}}$;
		\ENDFOR
       \STATE $FacMatPart_1$ $\leftarrow$ $\textbf{X}_{i_n, j}\bm{GS}^{(n)}_{:}$ of Part (1) in Equ. (\ref{Gradient_low_rank});
       \STATE $FacMatPart_2$ $\leftarrow$ $\lambda_{a_{i_n, :}^{(n)}}a_{i_n, :}^{(n)}$ of Part (2) in Equ. (\ref{Gradient_low_rank});
       \STATE $InterMX$ $\leftarrow$ $a_{i_n, :}^{(n)}\bm{GS}^{(n)}_{:}$ of Part (3) in Equ. (\ref{Gradient_low_rank});
       \STATE $FacMatPart_3$ $\leftarrow$ $InterMX\bm{GS}^{(n)T}_{:}$ of Part (3) in Equ. (\ref{Gradient_low_rank});
       \STATE Update $a_{i_n, :}^{(n)}$ by SGD in Equ. (\ref{SGD}) after summing the parts (1)-(3) of Equ. (\ref{Gradient_low_rank});
	\ENDFOR
	\FOR{$n$ from $1$ to $N$}
        \STATE $j\leftarrow 1+\sum_{k=1,n\neq k}^{N}$ $\left[(i_{k}-1)\mathop{\prod}_{m=1,m\neq n}^{k-1}I_{m}\right]$;
        \STATE Set all $\textbf{Q}^{(n), r}_{:, j}\leftarrow 0$, $CoreTensorPart_4\leftarrow 0$, $n\in\{N\}$, $r\in\{R_{core}\}$;
	    \FOR{$r$ from $1$ to $R_{core}$}
           \FOR{$r_1$ from $1$ to $R_{core}$}
		      \FOR{$n_{0}$ from $1$ to $N$, $n_{0}$ $\neq$ $n$}
                 \STATE $c_{r_1}^{(n_{0})}\leftarrow \textbf{b}^{(n_{0})T}_{:, r_1}\textbf{a}^{(n_{0})T}_{i_{n_{0}}, :}$; \%This step can be accelerated by CUDA Warp Shuffle and Memory Coalescing.
                 \STATE $\textbf{a}^{(n_{0})}_{i_{n_{0}}, :}\leftarrow c_{r_1}^{(n_{0})}\textbf{a}^{(n_{0})}_{i_{n_{0}}, :}$;
		      \ENDFOR
              \STATE $\textbf{Q}^{(n_{0}), r_1}_{:, j}\leftarrow \textbf{a}^{(n_{0})}_{i_{n_{0}}, :}$;
		    \ENDFOR
           \FOR{$r_{core}$ from $1$ to $R_{core}$}
               \STATE $InterMX_{r_{core}}\leftarrow \textbf{b}^{(n)T}_{:, r_{core}}\textbf{Q}^{(n), r_{core}}_{:, j}$;
		   \ENDFOR
           \STATE $CoreTensorPart_1\leftarrow \textbf{X}_{i_n, j}\textbf{Q}^{(n), r}_{:, j}$;
           \STATE $CoreTensorPart_2\leftarrow \lambda_{\textbf{b}^{(n)}_{:,r}}\textbf{b}^{(n)T}_{:,r}$;
           \STATE $CoreTensorPart_3\leftarrow InterMX_{r}\textbf{Q}^{(n), rT}_{:, j}$;
           \FOR{$r_{core}$ from $1$ to $R_{core}$, $r_{core}$ $\neq$ $r$}
                \STATE $CoreTensorPart_4$ $\leftarrow$ $CoreTensorPart_4$ $+$ $InterMX_{r_{core}}\textbf{Q}^{(n), rT}_{:, j}$;
		   \ENDFOR
            \STATE Update $\textbf{b}^{(n)T}_{:,r})$ by SGD in Equ. (\ref{SGD}) after summing the parts (1)-(4) of Equ. (\ref{Gradient_low_rank_core});
		\ENDFOR
	\ENDFOR
		\STATE $\textbf{Return}$: $\textbf{A}^{(n)}$, $\textbf{B}^{(n)}$, $n$ $\in$ $\{N\}$.\\
	\end{algorithmic}
\end{algorithm}

\subsection{Complexity Analysis and Comparison}\label{section33}
The computational details and computational complexity of the proposed model are concluded in Table \ref{Computational_Complexity} and Algorithm \ref{sgd_sftd_a}, respectively.
This section also presents the computational complexity of the condition without the Kruskal product for approximating the core tensor
(In experimental section, we refer the cuTucker as the stochastic strategy for STD without Kruskal product for approximating the core tensor on GPU CUDA programming).

(1) With the Theorems \ref{theorem1} and \ref{theorem2},
the computational complexity of $\textbf{D}^{(n)}_{(\Psi^{(n)}_{M})_{i_n}, j}=\widehat{\bm{G}}^{(n)}\textbf{S}^{(n)T}_{(\Psi^{(n)}_{M})_{i_n}, j}$ can be reduced from direct computation $O\big(\prod_{k=1\neq n}^{N}J_k\big)$ to $DCC^{(n)}$ $=$ $O\big(R_{core}J_{n}\sum_{k=1\neq n}^{N}J_k\big)$.
The computational complexity and space overhead of $\textbf{Q}^{(n), rT}_{:, j}$ can be reduced from $O\big(J_{n}\sum_{k=1}^{N}J_k\big)$ into $QCC^{(n)}$ $=$ $O\big(J_{n}\sum_{k=1\neq n}^{N}J_k\big)$.

(2) Without Kruskal approximation,
in the process of updating $\textbf{a}^{(n)}_{i_{n}}$, $i_n\in \{I_{n}\}$, $n \in\{N\}$,
the computational complexity of the intermediate matrices $\big\{\textbf{S}^{(n)T}_{(\Psi^{(n)}_{M})_{i_n}, j}$ $=$
$\textbf{a}^{(N)}_{i_N, :}\otimes\cdots \otimes\textbf{a}^{(n+1)}_{i_{n+1}, :}\otimes\textbf{a}^{(n-1)}_{i_{n-1}, :}\otimes\cdots \otimes\textbf{a}^{(1)}_{i_1, :}$ $,$
$\textbf{D}^{(n)}_{(\Psi^{(n)}_{M})_{i_n}, j}$
$=$
$\bm{G}^{(n)}$
$\textbf{S}^{(n)}_{(\Psi^{(n)}_{M})_{i_n}, j}$
$\big\}$, where $j=1+\sum_{k=1,n\neq k}^{N}$ $\left[(i_{k}-1)\mathop{\prod}_{m=1,m\neq n}^{k-1}I_{m}\right]$, $(i_{1},\cdots,i_{n},\cdots,i_{N})$ $\in$ $\Psi$ is
$\big\{$ $O\big(\prod_{k=1\neq n}^{N}J_k\big)$ $,$ $O\big(\prod_{k=1}^{N}J_k\big)$ $\big\}$, respectively, and
$O\big(\prod_{k=1}^{N}J_k\big)$ $\gg$ $DCC^{(n)}$.
In the process of updating $g^{(n)}$,
the space overhead complexity of coefficient matrix
$\big\{$
$\textbf{H}^{(n)}_{j, :}$ $\in$ $\mathbb{R}^{J_{N} \cdots J_{n+1} J_{n-1} \cdots J_{1} J_{n} }$ $\big| n\in\{N\}$ $j=1+\sum_{k=1}^{N}$ $\left[(i_{k}-1)\mathop{\prod}_{m=1,m\neq n}^{k-1}I_{m}\right]$, $(i_{1},\cdots,i_{n},\cdots,i_{N})$ $\in$ $\Psi$
$\big\}$
is
$O\big(\prod_{n=1}^{N}J_n\big)$,
and the computational complexity of the intermediate matrices
$\big\{$
$\textbf{H}^{(n)}_{j, :}$ $,$
$\textbf{H}^{(n)}_{j, :}g^{(n)}\big|n\in\{N\}$
$\big\}$
are $\big\{\mathop{\prod}_{m=1}^{N}J_{m},\mathop{\prod}_{m=1}^{N}J_{m}\big\}$ $\gg$ $QCC^{(n)}$, respectively.

According to Table \ref{Computational_Complexity} and above analysis,
we can conclude that the proposed model can reduce the exponential overhead into linear one with the Theorems \ref{theorem1} and \ref{theorem2} and Kruskal approximation strategy.

\section{cuFastTucker on GPUs} \label{section4}
The Section \ref{section3} solves the problem of high computational overhead for updating the factor matrix and core tensor of STD with Kruskal
approximation and Theorems \ref{theorem1} and \ref{theorem2}.
However, the STD for the HOHDST data still replies on the modern HPC resource to obtain the real-time result.
Due to the basic computational part of thread and thread block and fine-grained and high parallelization of the proposed model,
the GPU is chosen to further accelerate the proposed model (cuFastTucker).
The parallelization strategy is divided into two parts:
(1) thread parallelization within a thread block;
(2) parallelization of thread block.
In this section, data partition and communication on multi-GPUs are also presented.

\begin{figure}[htbp]
	\centering
	\includegraphics[width=3.4in]{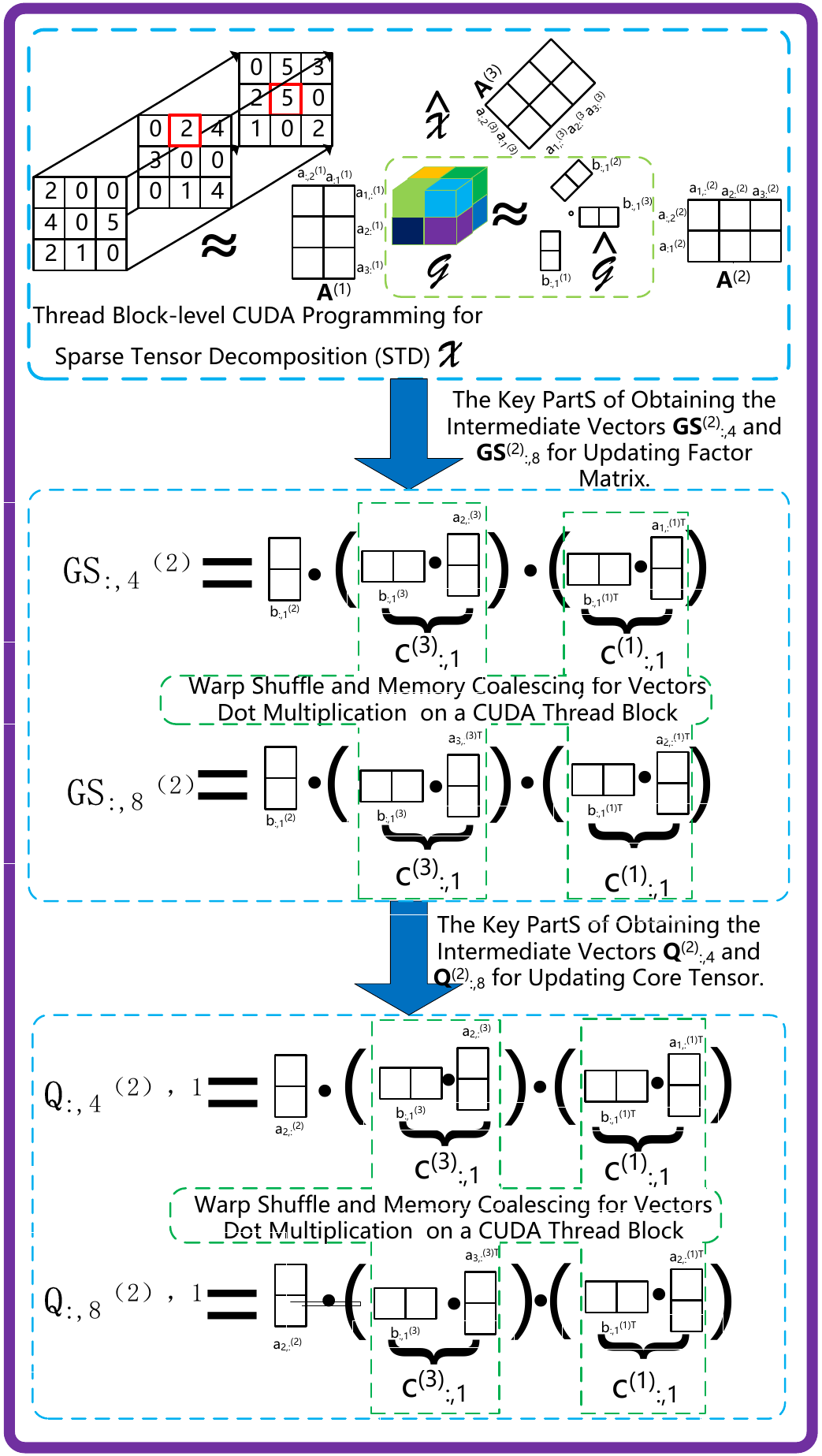}
	\caption{Two key steps $\{\bm{GS}^{(n)}_{:,j}, \textbf{Q}^{(n), r}_{:, j}\}$ for updating the factor matrix and core tensor in a CUDA thread block.}
	\label{Two_parts_CUDA}
\end{figure}

\subsection{CUDA Thread Parallelization within a Thread Block} \label{section41}
GPUs is a Single Instruction Multiple Devices architecture, where a thread block can be packed as a thread group,
and the current size of scheduling unit (Warp) in CUDA GPUs is $32$.
Hence, the number of threads within a thread block are a divisor or multiple of $32$.
This means that when $J_n$, $n \in \{N\}$ is a divisor or multiple of 32, it has better performance.
Fig. \ref{Two_parts_CUDA} illustrates the two key steps $\{\bm{GS}^{(n)}_{:,j}, \textbf{Q}^{(n), r}_{:, j}\}$ for updating the factor matrix and core tensor in a CUDA thread block.
The major optimization techniques in cuFastTucker are concluded as:

\textbf{Warp Shuffle}:
as the Fig. \ref{Two_parts_CUDA} illustrate,
warp shuffle instructions in cuFastTucker are used to compute the dot product or sum operations $\{\bm{GS}^{(n)}_{:,j}, \textbf{Q}^{(n), r}_{:, j}\}$ (Lines 6 and 23 in Algorithm \ref{sgd_sftd_a}) and then broadcast the result, automatically.
The warp shuffle instruction needs additional hardware support, with lower latency and no additional memory resources,
and it allows a thread to directly read the register values of other threads in the same thread warp,
which has better communication efficiency than reading and writing data through the shared memory.

\textbf{On-chip Cache}:
the current GPUs allows programmers to control the caching behavior of each memory instruction of the on-chip L1 cache.
In the SGD based method, the read-only index and read-only value of the non-zero element may be frequently reused in the near future (temporal reuse) or by other thread blocks (spatial reuse).
We use $\_\_ldg$ to modify read-only memory to improve access efficiency.

\textbf{Memory Coalescing}:
in order to utilize the bus bandwidth, GPUs usually coalesce the memory accesses of multiple threads into fewer memory requests.
Assuming that the multiple threads access addresses within $32$ bytes,
their access can be completed through a memory request, which can greatly improve bandwidth utilization.
According to the CUDA code, all the variable matrices
$\big\{$ $\textbf{A}^{(n)}$,
$\textbf{G}^{(n)}$,
$\textbf{B}^{(n)}$ $\big|$
$n$ $\in$ $\{N\}$ $\big\}$ are stored as the form of $\big\{$ $\textbf{A}^{(n)}$ $\in$ $\mathbb{R}^{I_{n}\times J_{n}}$,
$\textbf{G}^{(n)^{T}}$,
$\textbf{B}^{(n)^{T}}$ $\in$ $\mathbb{R}^{R_{core}\times J_{n}}$ $\big|$
$n$ $\in$ $\{N\}$ $\big\}$ to ensure that consecutive threads access consecutive memory addresses.

\textbf{Register Usage}:
the register file is the fastest storage unit on GPUs,
so we save every reusable variable in registers.
Although the total number of registers on GPUs is fixed,
our algorithm only needs to use a small number of registers.
The current number of GPUs is completely sufficient.

\textbf{Shared Memory}:
the threads within a thread block use the shared memory which is much faster than global memory,
and the register is not suitable for storing continuous vectors.
The frequent used intermediate vectors lie in shared memory,
and these vectors will be used in the next process.

\subsection{CUDA Thread Block Parallelization} \label{section42}
From Algorithm \ref{sgd_sftd_a},
the computational step for each feature vector $a^{(n)}_{i_n, :}$, $i_n \in \{I_{n}\}$, $n\in\{N\}$ is independent which has fine-grain parallelization.
Meanwhile, Kruskal approximation vectors $b^{(n)}_{:,r}$, $n\in\{N\}$, $r\in \{R_{core}\}$ are dependent.
Thus, the vectors $b^{(n)}_{:,r}$, $n\in\{N\}$, $r\in \{R_{core}\}$ should be updated, simultaneously.
The allocated $NR_{core}$ thread blocks pre-compute the gradient following the line 23 of Algorithm \ref{sgd_sftd_a}, and then
update $b^{(n)}_{:,r}$, $n\in\{N\}$, $r\in \{R_{core}\}$, simultaneously.
There are two strategies to allocate the $TB$ thread blocks to update the feature vector $a^{(n)}_{i_n, :}$, $i_n \in \{I_{n}\}$, $n\in\{N\}$:
(1) the $I_{n}$, $n\in\{N\}$ feature vectors are allocated to the $TB$ thread blocks;
(2) each thread block selects a index from the one-step sampling set $\Psi$ and each thread block compute the gradient following the line 6 of Algorithm \ref{sgd_sftd_a}.

\subsection{Workload Partitioning} \label{section43}

The scale of data that can be processed by a single GPU is limited, and
otherwise, the time required to process the data is not acceptable.
The HOHDST data is divided so that it can be processed simultaneously on multiple GPUs.
Given an $N$-order tensor $\mathcal{X}\in\mathbb{R}^{I_{1}\times\cdots I_{n}\cdots \times I_{N}}$ and $M$ GPUs,
we averagely cut each order of the tensor into $M$ parts,
so the tensor is evenly divided into $M^N$ blocks $\{\mathcal{X}_{1,1,\cdots,1},$
$\mathcal{X}_{1,1,\cdots,2},\cdots,$
$\mathcal{X}_{M,M,\cdots,M-1},$
$\mathcal{X}_{M,M,\cdots,M}\}$.
In the same period of time,
each GPU is responsible for processing one of these $M^N$ blocks.
In order to avoid data conflicts between multiple GPUs,
the indexes of the same order of the blocks that are responsible for different GPUs at the same time are different.
For example, block $\mathcal{X}_{1,1,\cdots,1}$ and block $\mathcal{X}_{2,2,\cdots,1}$ cannot be processed at the same time.

As shown in Figure \ref{Multi-GPU}, 
two GPUs are used to process a 3-dimensional tensor, 
which is equally divided into $2*2*2$ blocks. 
When updating the factor matrices, GPU 1 and GPU 2 update 
$\big\{\{A_{1}^{(1)},$ $A_{1}^{(2)},$ $A_{1}^{(3)}\},$
$\{A_{1}^{(1)},$ $A_{1}^{(2)},$ $A_{2}^{(3)}\},$
$\{A_{1}^{(1)},$ $A_{2}^{(2)},$ $A_{2}^{(3)}\},$
$\{A_{1}^{(1)},$ $A_{2}^{(2)},$ $A_{1}^{(3)}\}\big\}$ and 
$\big\{\{A_{2}^{(1)},$ $A_{2}^{(2)},$ $A_{2}^{(3)}\},$
$\{A_{2}^{(1)},$ $A_{2}^{(2)},$ $A_{1}^{(3)}\},$
$\{A_{2}^{(1)},$ $A_{1}^{(2)},$ $A_{1}^{(3)}\},$
$\{A_{2}^{(1)},$ $A_{1}^{(2)},$ $A_{2}^{(3)}\}\big\}$
through blocks 
$\big\{Tensor(1,1,1),$ $Tensor(1,1,2),$ $Tensor(1,2,2),$ $Tensor(1,2,1)\big\}$ 
and $\big\{Tensor(2,2,2),$ $Tensor(2,2,1),$ $Tensor(2,1,1),$ $Tensor(2,1,2)\big\}$, respectively. 
In this process, the processing of GPU 1 and GPU 2 does not conflict, and after GPU 1 and GPU 2 update blocks $\big\{Tensor(1,1,1),$ $Tensor(1,1,2),$ $Tensor(1,2,2),$ $Tensor(1,2,1)\big\}$ 
and $\big\{Tensor(2,2,2),$ $Tensor(2,2,1),$ $Tensor(2,1,1),$ $Tensor(2,1,2)\big\}$ respectively,
they only need to pass parameters $\big\{$ $A_{1}^{(3)},$
$A_{1}^{(2)},$
$A_{2}^{(3)},$
$A_{2}^{(2)}\big\}$ and 
$\big\{A_{2}^{(3)},$
$A_{2}^{(2)},$
$A_{1}^{(3)},$
$A_{1}^{(2)}\big\}$ to each other.
When updating the core tensor, 
it is only necessary to update the core tensor after accumulating all the gradients.

\begin{figure}[htbp]
	\centering
	\includegraphics[width=3.4in]{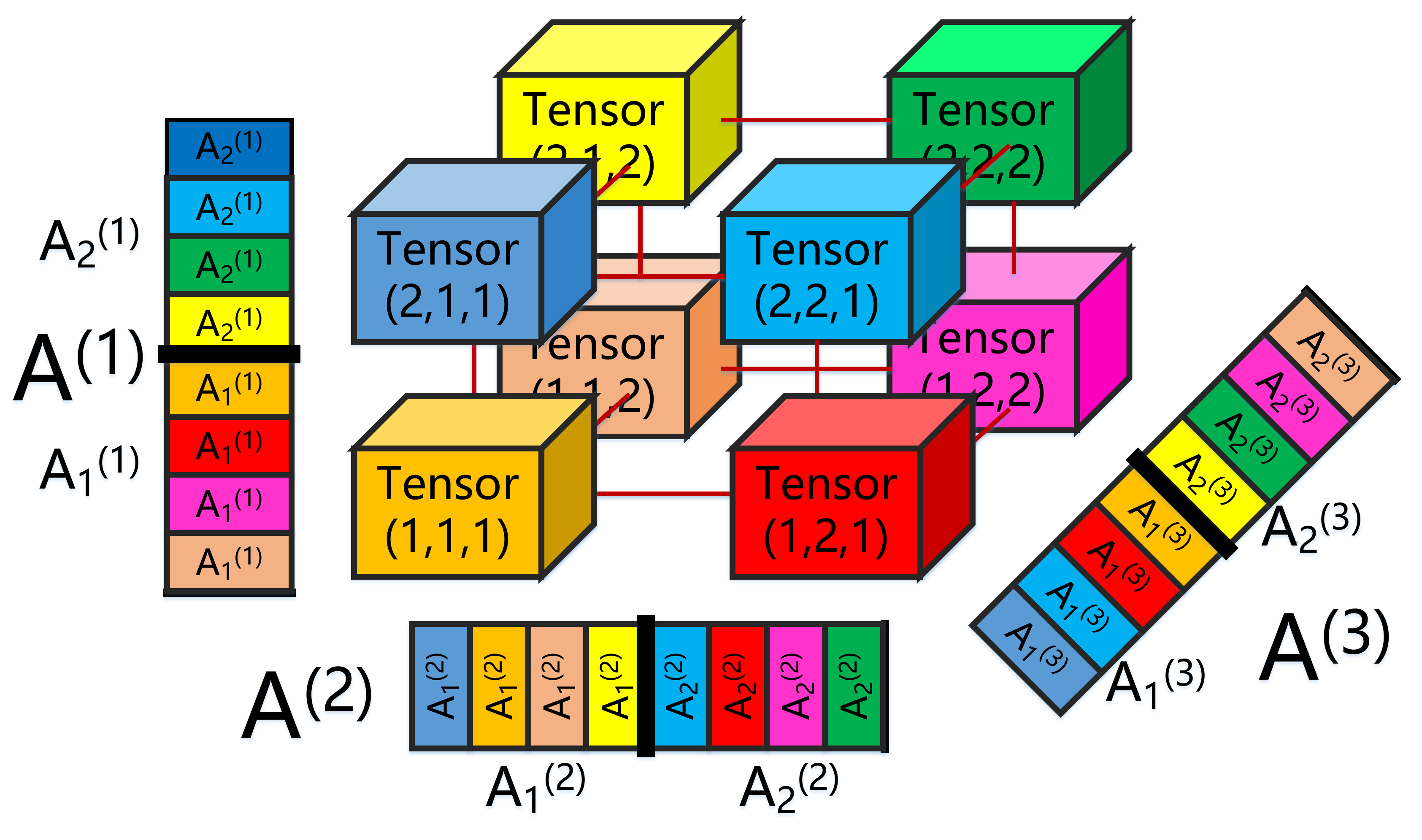}
	\caption{An example of workload partitioning of a tensor.}
	\label{Multi-GPU}
\end{figure}

\begin{table}[htbp]
	\centering
	\footnotesize
	\setlength{\abovecaptionskip}{0pt}
	\caption{Real World Datasets}
	\begin{tabular}{c|ccc}
		\hline
		\hline
						& Netflix        & Yahoo!Music      & Amazon Reviews    \\
		\hline
		$I_1$           & 480, 189       & 1, 000, 990      & 4, 821, 207       \\
		$I_2$           & 17, 770        & 624, 961         & 1, 774, 269       \\
		$I_3$           & 2, 182         & 3, 075           & 1, 805, 187       \\
		$|\Omega|$      & 99, 072, 112   & 250, 272, 286    & 1, 741, 809, 018  \\
		$|\Gamma|$      & 1, 408, 395    & 2, 527, 989      & -                 \\
		Max Value       & 5              & 5                & -                 \\
		Min Value       & 1              & 0.025            & -                 \\
		\hline
		\hline
	\end{tabular}
	\label{data_sets_real}
\end{table}

\begin{table}[htbp]
	\centering
	\scriptsize
	\setlength{\abovecaptionskip}{0pt}
	\caption{Synthesis Datasets}
	\begin{tabular}{c|cccc}
		\hline
		\hline
						& Order-3            & Order-4         & Order-5        & Order-6 to Order-10      \\
		\hline
		$I$             & 10, 000            & 10, 000         & 10, 000        & 10, 000                  \\
		$|\Omega|$      & 1G                 & 800M            & 600M           & 100M                     \\
		Max Value       & 5                  & 5               & 5              & 5                        \\
		Min Value       & 1                  & 1               & 1              & 1                        \\
		\hline
		\hline
	\end{tabular}
	\label{data_sets_simulation}
\end{table}

\begin{table}[htbp]
	\centering
	\tiny
	\setlength{\abovecaptionskip}{0pt}
	\caption{The initial learning rate and regularization parameters of cuTucker on Netflix and Yahoo!Music datasets.}
	\begin{tabular}{c|cccc|cccc}
		\hline
		\hline
					  & \multicolumn{4}{c|}{Netflix}                & \multicolumn{4}{c}{Yahoo!Music}              \\
		\hline
		$J$           & 4       & 8         & 16        & 32        & 4        & 8         & 16        & 32        \\
		$\alpha_a$    & 0.0060  & 0.0045    & 0.0025    & 0.0005    & 0.0045   & 0.0040    & 0.0025    & 0.0005    \\
		$\beta_a$     & 0.05    & 0.05      & 0.05      & 0.05      & 0.2      & 0.2       & 0.2       & 0.2       \\	
		$\lambda_a$   & 0.01    & 0.01      & 0.01      & 0.01      & 0.01     & 0.01      & 0.01      & 0.01      \\
		$\alpha_g$    & -       & 0.0045    & 0.0035    & 0.0025    & -        & 0.0045    & 0.0035    & 0.0025    \\
		$\beta_g$     & -       & 0.1       & 0.1       & 0.1       & -        & 0.1       & 0.1       & 0.1       \\
		$\lambda_g$   & -       & 0.01      & 0.01      & 0.01      & -        & 0.01      & 0.01      & 0.01      \\
		\hline
		\hline
	\end{tabular}
	\label{the_parameters_of_cuTucker}
\end{table}

\begin{table}[htbp]
	\centering
	\tiny
	\setlength{\abovecaptionskip}{0pt}
	\caption{The initial learning rate and regularization parameters of cuFastTucker on Netflix and Yahoo!Music datasets.}
	\begin{tabular}{c|cccc|cccc}
		\hline
		\hline
					  & \multicolumn{4}{c|}{Netflix}            & \multicolumn{4}{c}{Yahoo!Music}             \\
		\hline
		$J$           & 4       & 8        & 16      & 32       & 4        & 8        & 16        & 32        \\
		$\alpha_a$    & 0.009   & 0.0060   & 0.0036  & 0.0020   & 0.0070   & 0.0060   & 0.035     & 0.0018    \\
		$\beta_a$     & 0.05    & 0.05     & 0.05    & 0.05     & 0.2      & 0.2      & 0.2       & 0.2       \\	
		$\lambda_a$   & 0.01    & 0.01     & 0.01    & 0.01     & 0.01     & 0.01     & 0.01      & 0.01      \\
		$\alpha_b$    & -       & 0.0045   & 0.0035  & 0.0025   & -        & 0.0045   & 0.0035    & 0.0025    \\
		$\beta_b$     & -       & 0.1      & 0.1     & 0.1      & -        & 0.1      & 0.1       & 0.1       \\
		$\lambda_b$   & -       & 0.01     & 0.01    & 0.01     & -        & 0.01     & 0.01      & 0.01      \\
		\hline
		\hline
	\end{tabular}
	\label{the_parameters_of_cuFastTucker}
\end{table}

\begin{table}[htbp]
	\centering
	\scriptsize
	\setlength{\abovecaptionskip}{0pt}
	\caption{The time (Seconds) overhead influence for updating the core tensor of cuTucker with the intermediate matrix on shared memory and global memory.}
	\begin{tabular}{c|cc|cc}
		\hline
		\hline
					    & \multicolumn{2}{c|}{Netflix}   & \multicolumn{2}{c}{Yahoo!Music}     \\
		\hline
		$J$             & 4             & 8              & 4                & 8                \\
		Shared Memory   & 0.274793      & 3.213333       & 0.707503         & 8.155324         \\
		Global Memory   & 0.382566      & 2.751087       & 1.227321         & 9.462652         \\	
		\hline
		\hline
	\end{tabular}
	\label{cuTucker_time}
\end{table}

\begin{table}[htbp]
	\centering
	\tiny
	\setlength{\abovecaptionskip}{0pt}
	\caption{The time (Seconds) overhead influence for updating the factor matrices of cuFastTucker with the core tensor on shared memory and global memory(NVIDIA Tesla P100 GPU).}
	\begin{tabular}{c|ccc|ccc}
		\hline
		\hline
						& \multicolumn{3}{c|}{Netflix}                  & \multicolumn{3}{c}{Yahoo!Music}               \\
		\hline
		$J$/$R_{core}$  & 4/4           & 8/4            & 8/8          & 4/4          & 8/4          & 8/8             \\
		Shared Memory   & 0.188097      & 0.415349       & 0.791291     & 0.634769     & 1.040703     & 1.989078        \\
		Global Memory   & 0.192428      & 0.411546       & 0.781091     & 0.672929     & 1.037247     & 1.969362        \\	
		\hline
		\hline
	\end{tabular}
	\label{cuTucker_time_a}
\end{table}

\begin{table}[htbp]
	\centering
	\tiny
	\setlength{\abovecaptionskip}{0pt}
	\caption{The time (Seconds) overhead influence for updating the core tensor of cuFastTucker with the core tensor on shared memory and global memory(NVIDIA Tesla P100 GPU).}
	\begin{tabular}{c|ccc|ccc}
		\hline
		\hline
			         	& \multicolumn{3}{c|}{Netflix}                   & \multicolumn{3}{c}{Yahoo!Music}              \\
		\hline
		$J$/$R_{core}$  & 4/4           & 8/4            & 8/8           & 4/4          & 8/4         & 8/8             \\
		Shared Memory   & 0.230835      & 0.501182       & 0.901919      & 0.630144     & 1.241826    & 2.248589        \\
		Global Memory   & 0.231354      & 0.509430       & 0.921666      & 0.642367     & 1.272383    & 2.314534        \\	
		\hline
		\hline
	\end{tabular}
	\label{cuTucker_time_b}
\end{table}

\begin{table}[htbp]
	\centering
	\tiny
	\setlength{\abovecaptionskip}{0pt}
	\caption{The time (Seconds) overhead influence for updating the factor matrices of cuFastTucker with the core tensor on shared memory and global memory(NVIDIA TITAN RTX GPU).}
	\begin{tabular}{c|ccc|ccc}
		\hline
		\hline
		& \multicolumn{3}{c|}{Netflix}                  & \multicolumn{3}{c}{Yahoo!Music}               \\
		\hline
		$J$/$R_{core}$  & 8/8           & 16/8           & 32/8         & 8/8          & 16/8         & 32/8            \\
		Shared Memory   & 0.294781      & 0.578777       & 1.217618     & 0.749726     & 1.462329     & 3.076224        \\
		Global Memory   & 0.247959      & 0.500432       & 1.083689     & 0.641186     & 1.269278     & 2.756746        \\	
		\hline
		\hline
	\end{tabular}
	\label{cuTucker_time_a_titan}
\end{table}

\begin{table}[htbp]
	\centering
	\tiny
	\setlength{\abovecaptionskip}{0pt}
	\caption{The time (Seconds) overhead influence for updating the core tensor of cuFastTucker with the core tensor on shared memory and global memory(NVIDIA TITAN RTX GPU).}
	\begin{tabular}{c|ccc|ccc}
		\hline
		\hline
		& \multicolumn{3}{c|}{Netflix}                   & \multicolumn{3}{c}{Yahoo!Music}              \\
		\hline
		$J$/$R_{core}$  & 8/8           & 16/8           & 32/8          & 8/8          & 16/8        & 32/8            \\
		Shared Memory   & 0.400963      & 0.737091       & 1.495815      & 1.032147     & 1.868572    & 3.792367        \\
		Global Memory   & 0.397872      & 0.711190       & 1.420703      & 1.029504     & 1.806532    & 3.613708        \\	
		\hline
		\hline
	\end{tabular}
	\label{cuTucker_time_b_titan}
\end{table}

\begin{figure*}[htbp]
  \centering
    \subfigure[RMSE on Netflix]{
    \label{fig501 (a1)} 
    \includegraphics[width=2.8in]{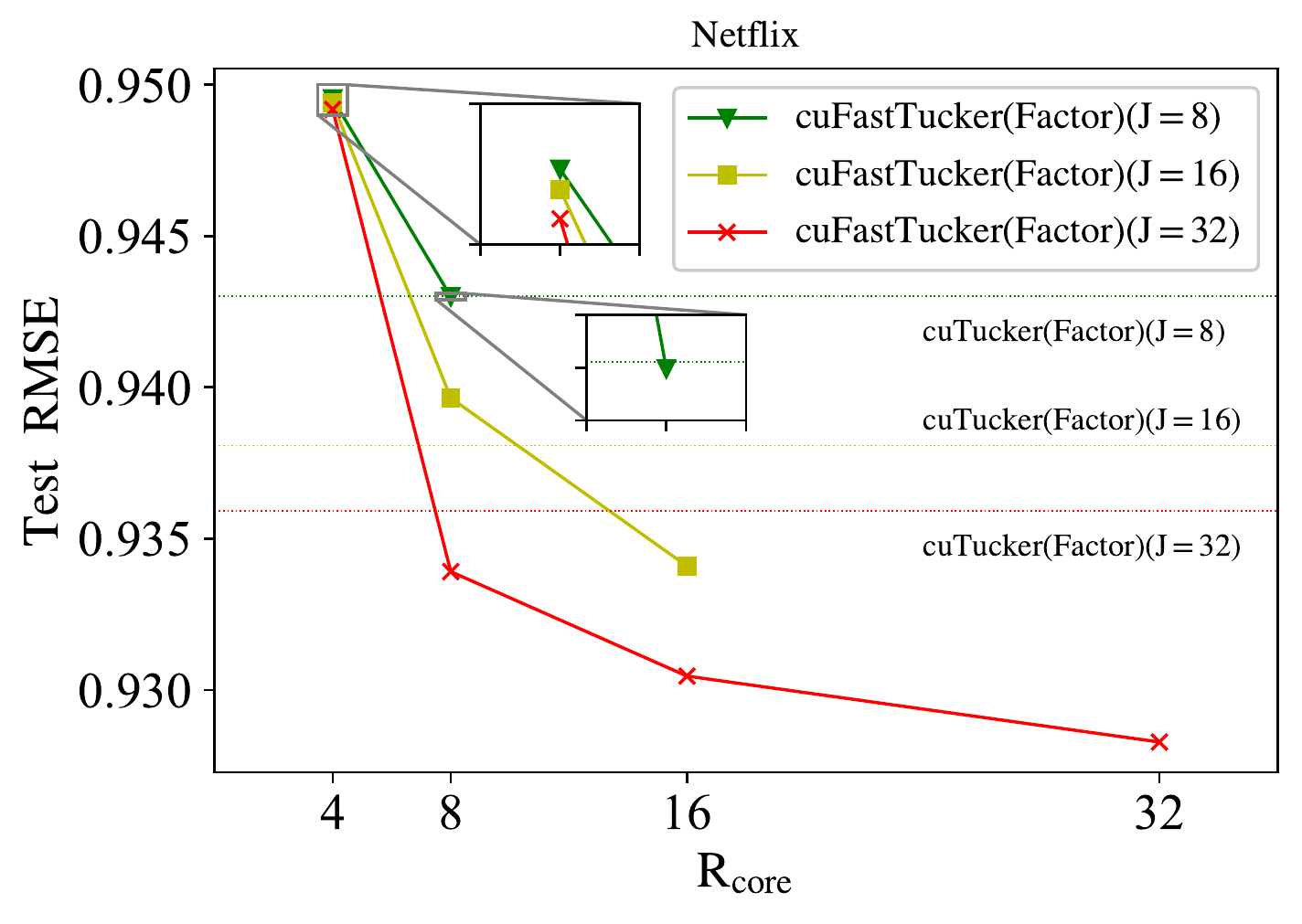}}
    ~
    \subfigure[MAE on Netflix]{
    \label{fig501 (a2)} 
    \includegraphics[width=2.8in]{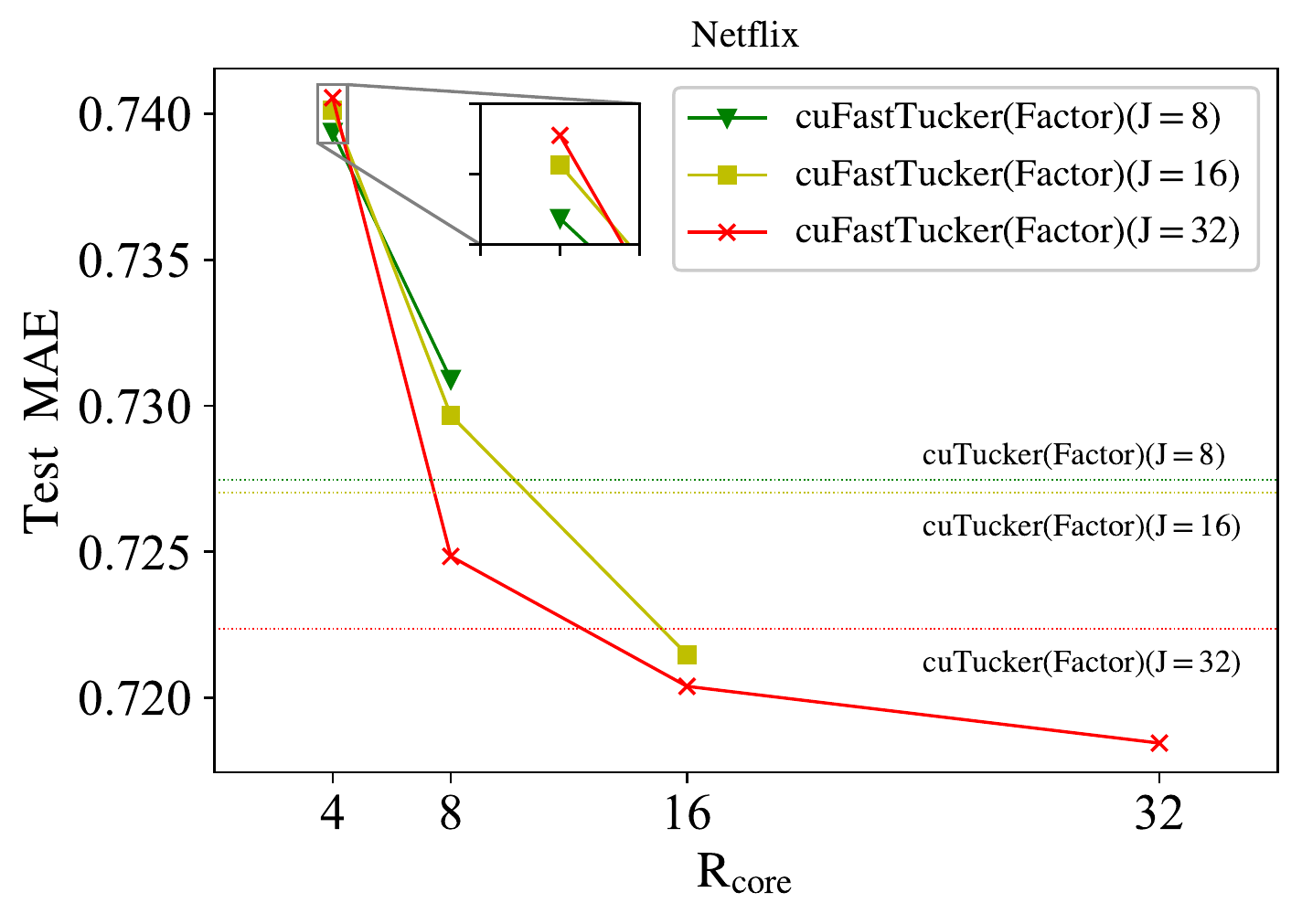}}
    ~
    \subfigure[RMSE on Yahoo!Music]{
    \label{fig501 (a3)} 
    \includegraphics[width=2.8in]{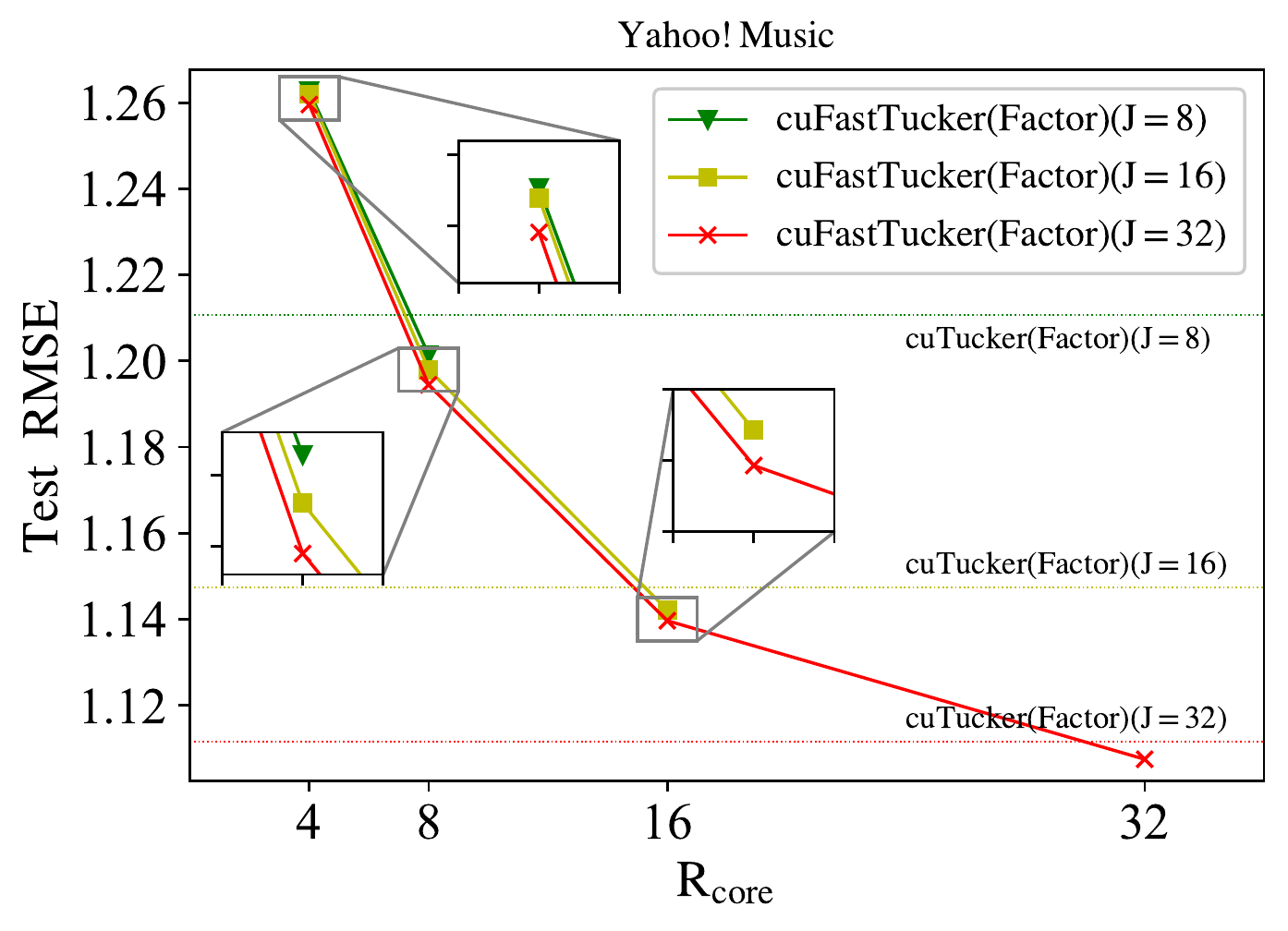}}
        ~
      \subfigure[MAE on Yahoo!Music]{
      \label{fig501 (a4)} 
    \includegraphics[width=2.8in]{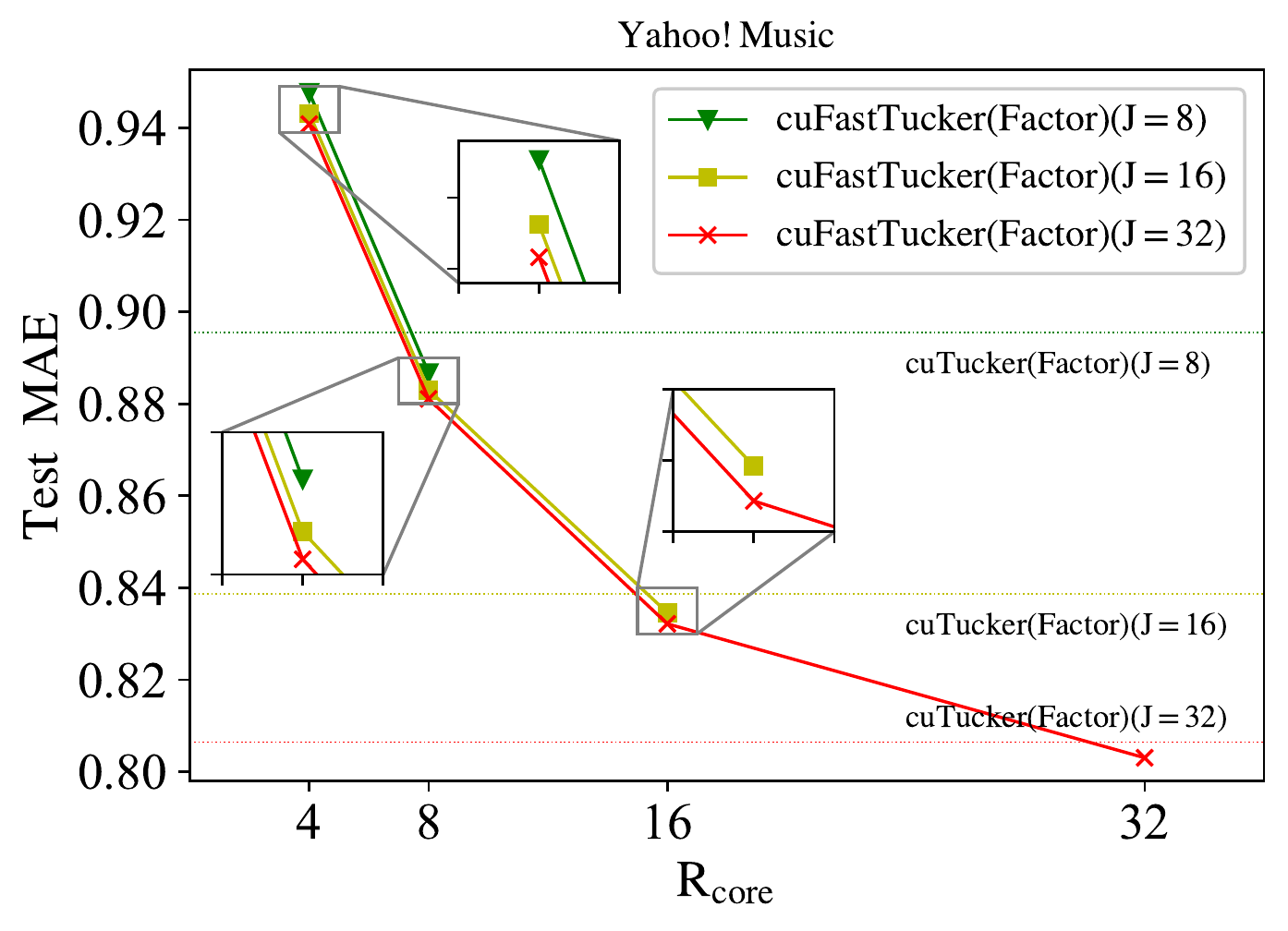}}
         \caption{The accuracy influence of the baseline algorithm (cuTucker) and cuFastTucker in various parameters $R_{core}=\{8, 16, 32\}$ with fixed $J_{n}=\{8, 16, 32\}$, $n\in\{N\}$.}
    \label{Impact_of_different_parameters}
\end{figure*}

\begin{figure*}[htbp]
  \centering
    \subfigure[RMSE on Netflix]{
    \label{fig502 (a1)} 
    \includegraphics[width=2.8in]{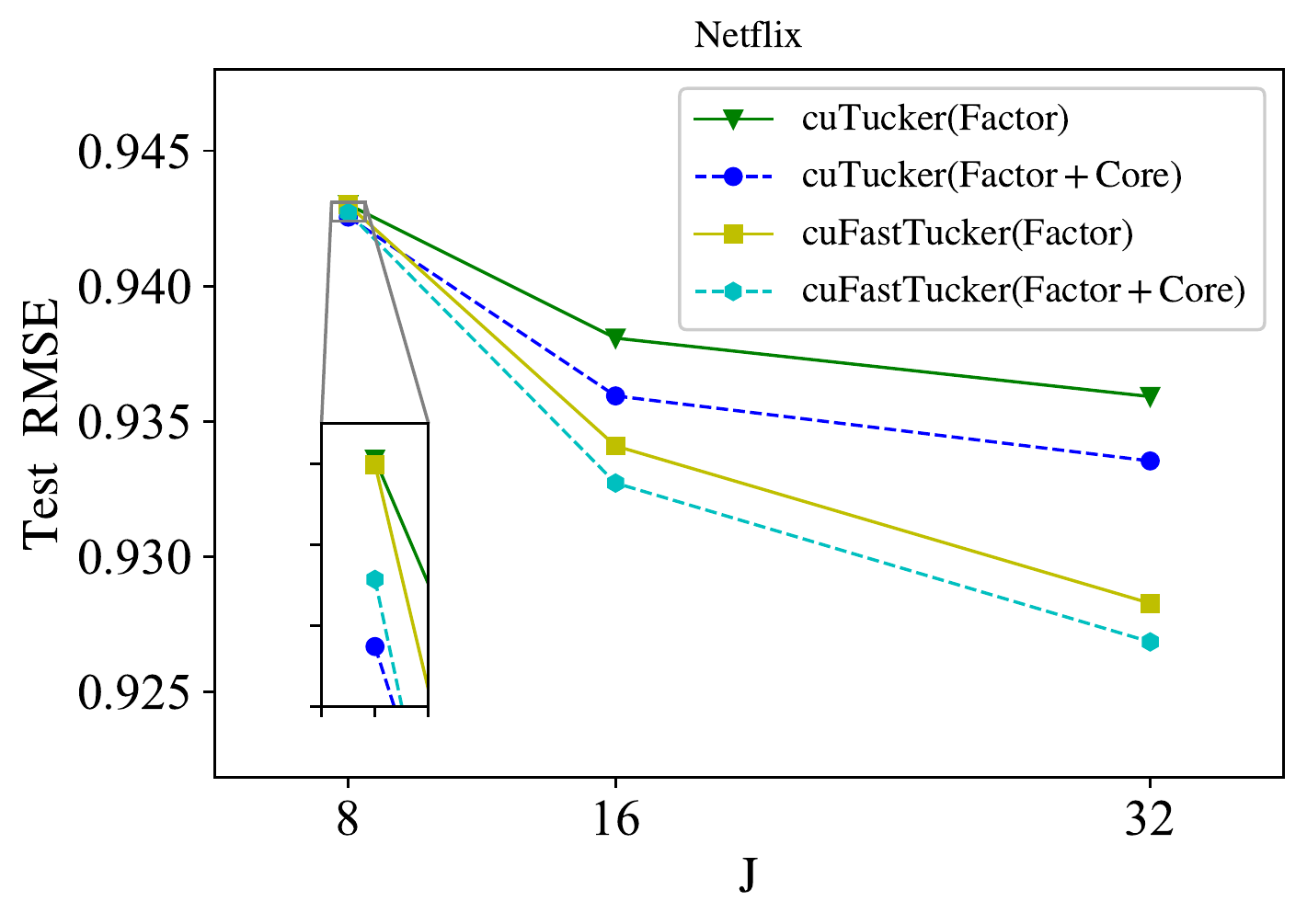}}
    ~
    \subfigure[MAE on Netflix]{
    \label{fig502 (a2)} 
    \includegraphics[width=2.8in]{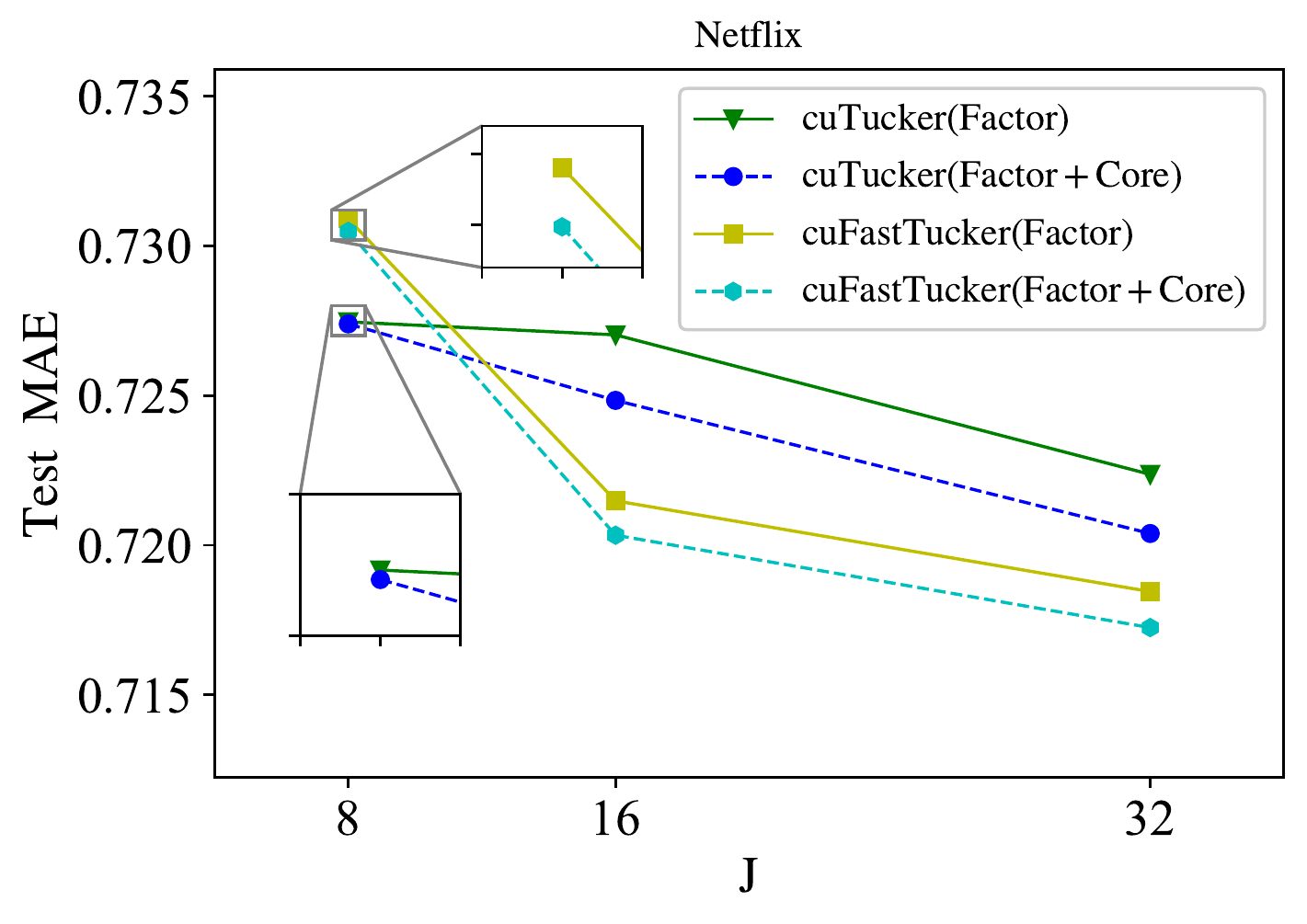}}
    ~
    \subfigure[RMSE on Yahoo!Music]{
    \label{fig502 (a3)} 
    \includegraphics[width=2.8in]{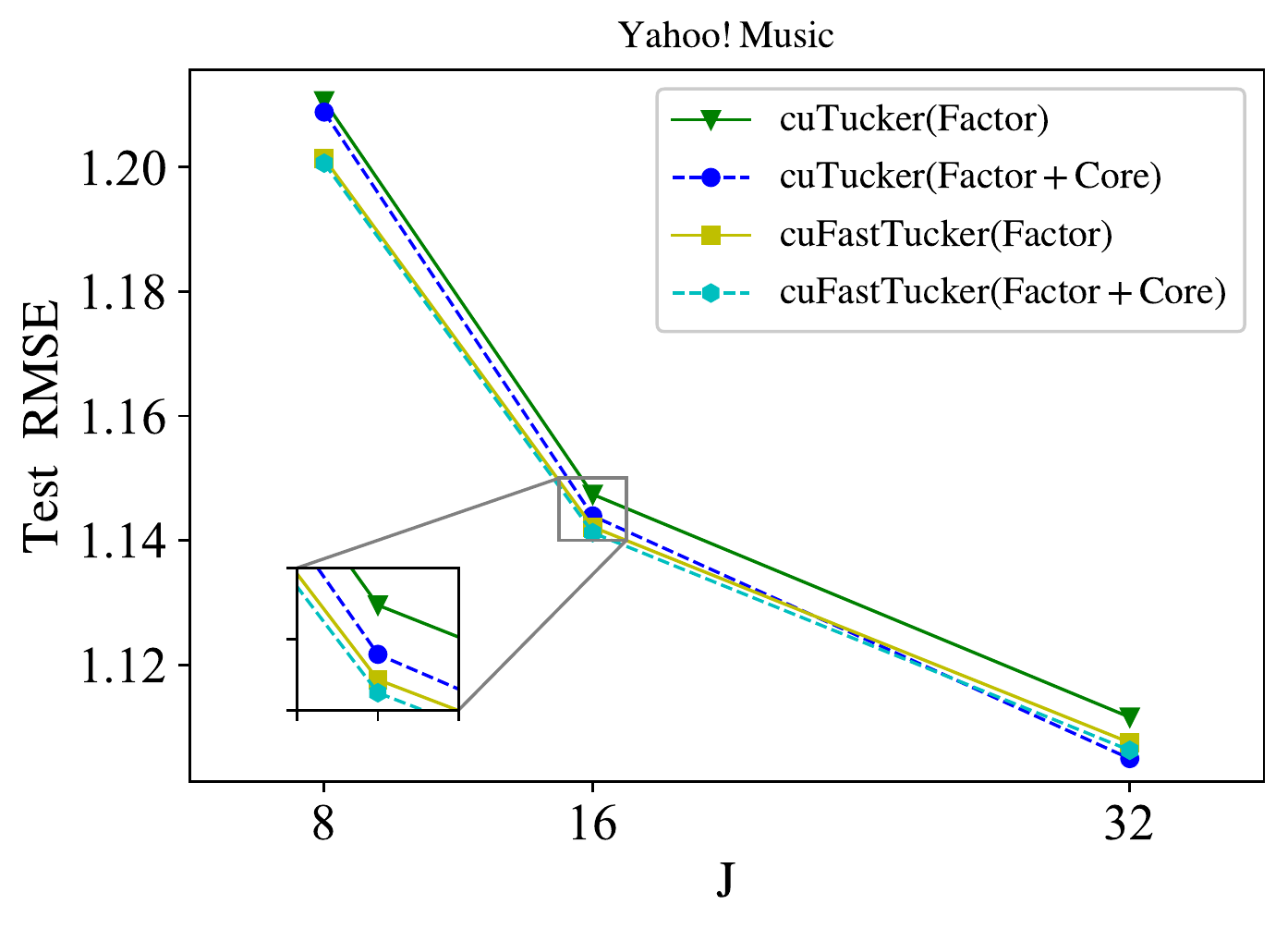}}
        ~
      \subfigure[MAE on Yahoo!Music]{
      \label{fig502 (a4)} 
    \includegraphics[width=2.8in]{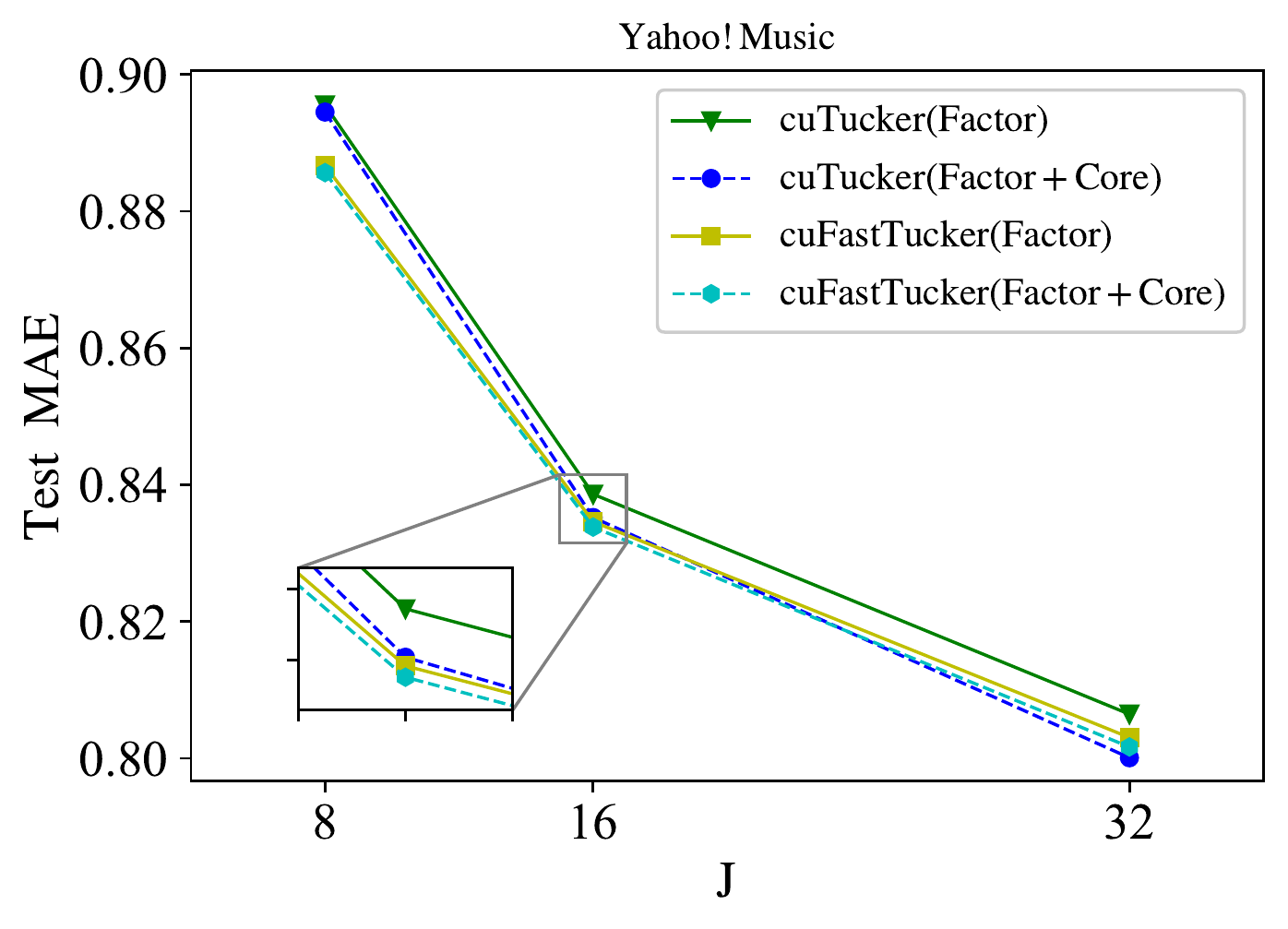}}
         \caption{The accuracy influence of the baseline algorithm (cuTucker) and cuFastTucker in various parameters $J_{n}=R_{core}=\{8, 16, 32\}$, $n\in\{N\}$.
         The experimental results illustrate two classes: (1) only updating factor matrix (Factor), (2) updating both the factor matrix and core tensor (Factor+Core).}
    \label{Impact_of_different_parameters_core}
\end{figure*}

\begin{figure*}
  \centering
    \subfigure[Time (Seconds) on Netflix]{
    \label{fig503 (a1)} 
    \includegraphics[width=2.8in]{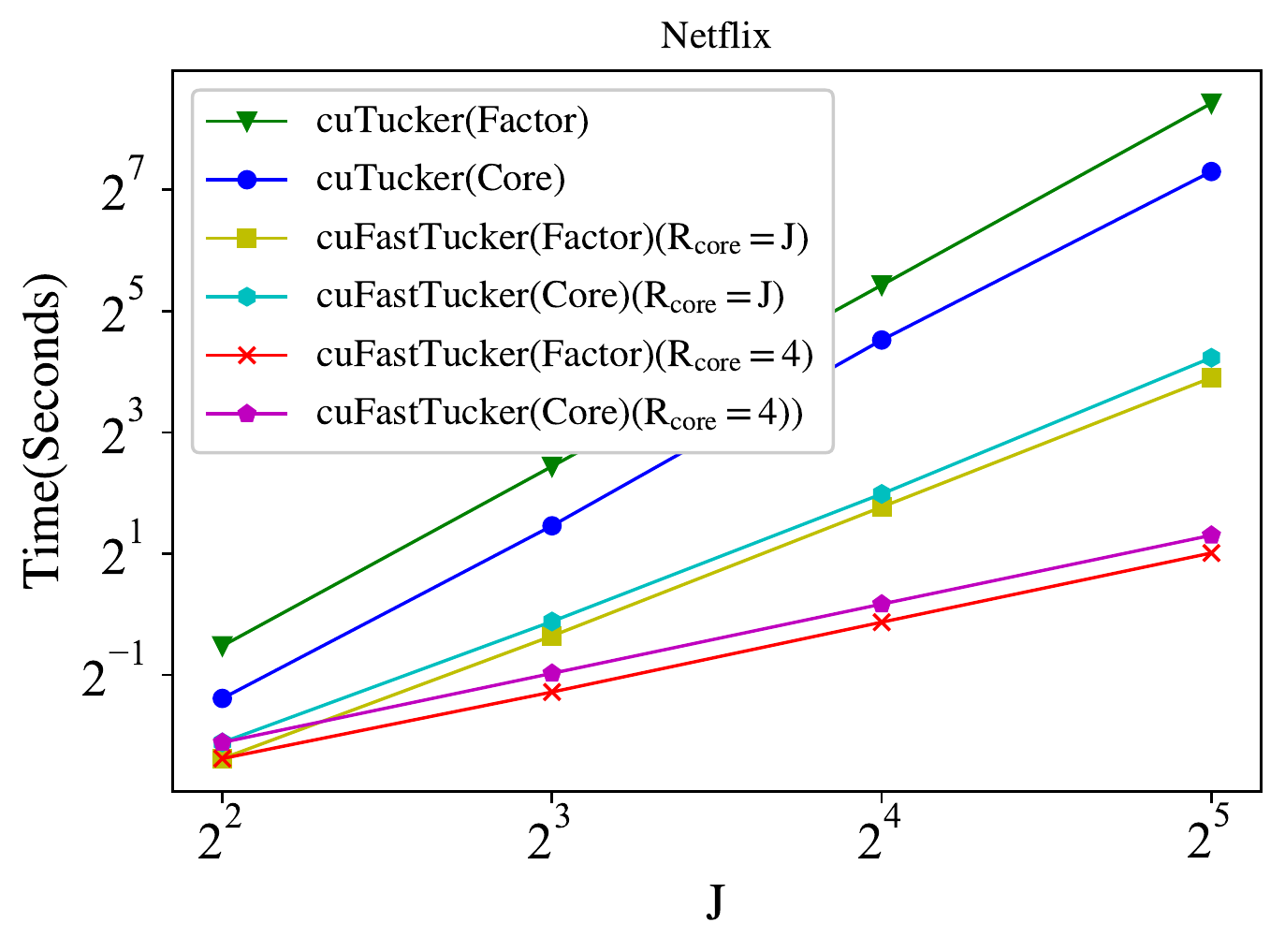}}
    ~
    \subfigure[Time (Seconds) on Yahoo!Music]{
    \label{fig503 (a2)} 
    \includegraphics[width=2.8in]{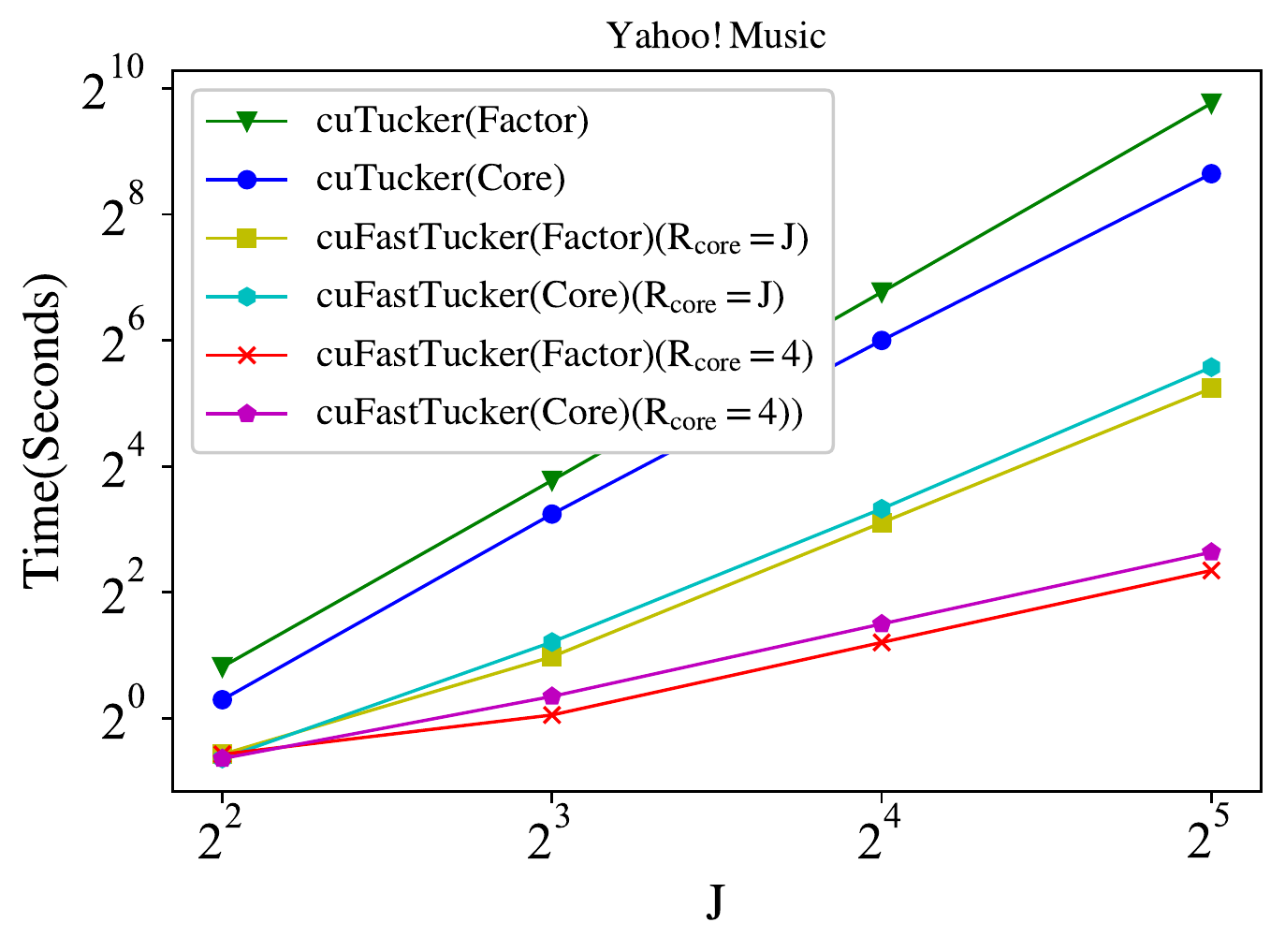}}
    ~
    \subfigure[Time (Seconds) on Netflix]{
    \label{fig503 (a3)} 
    \includegraphics[width=2.8in]{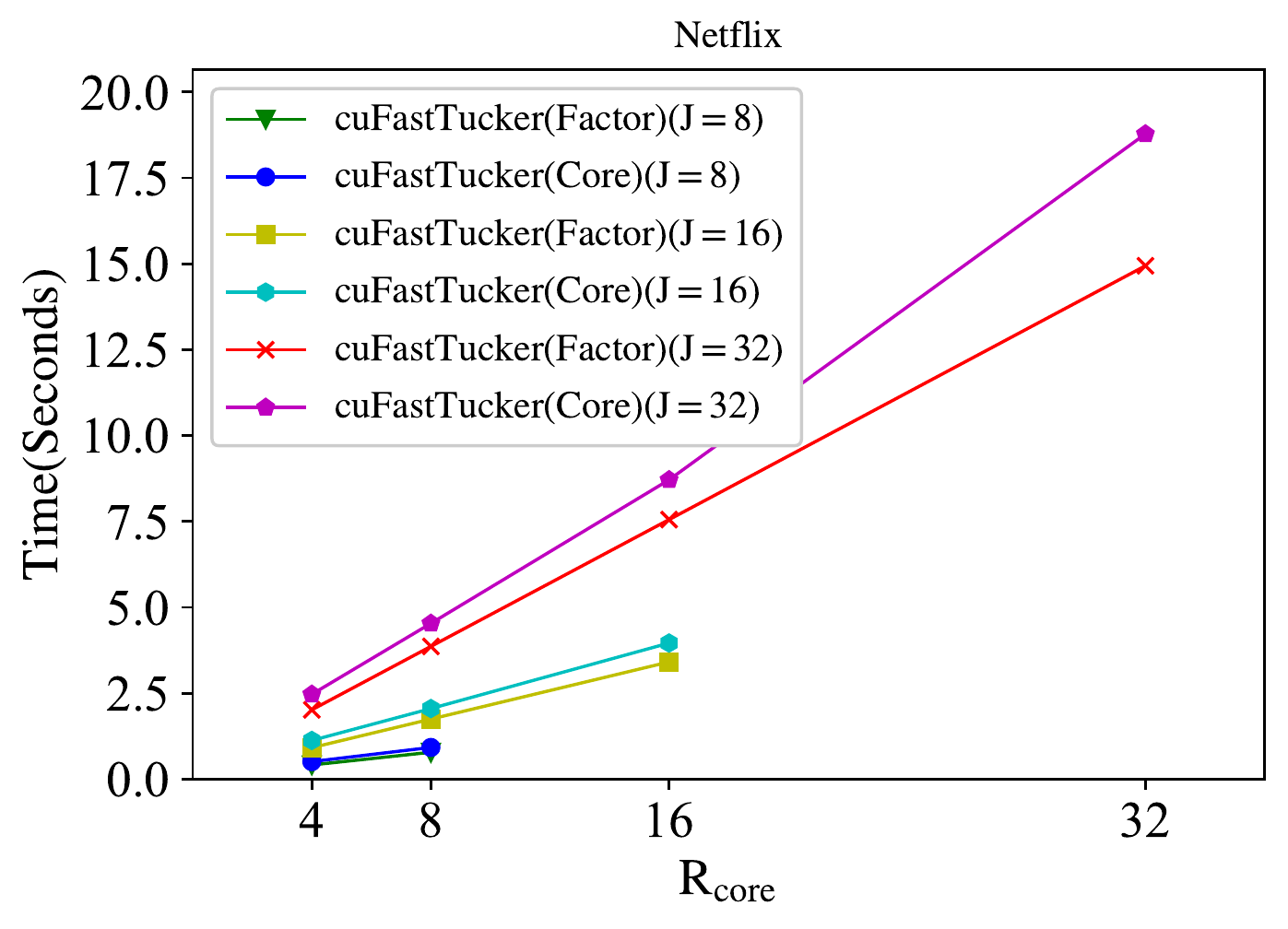}}
    ~
    \subfigure[Time (Seconds) on Yahoo!Music]{
    \label{fig503 (a4)} 
    \includegraphics[width=2.8in]{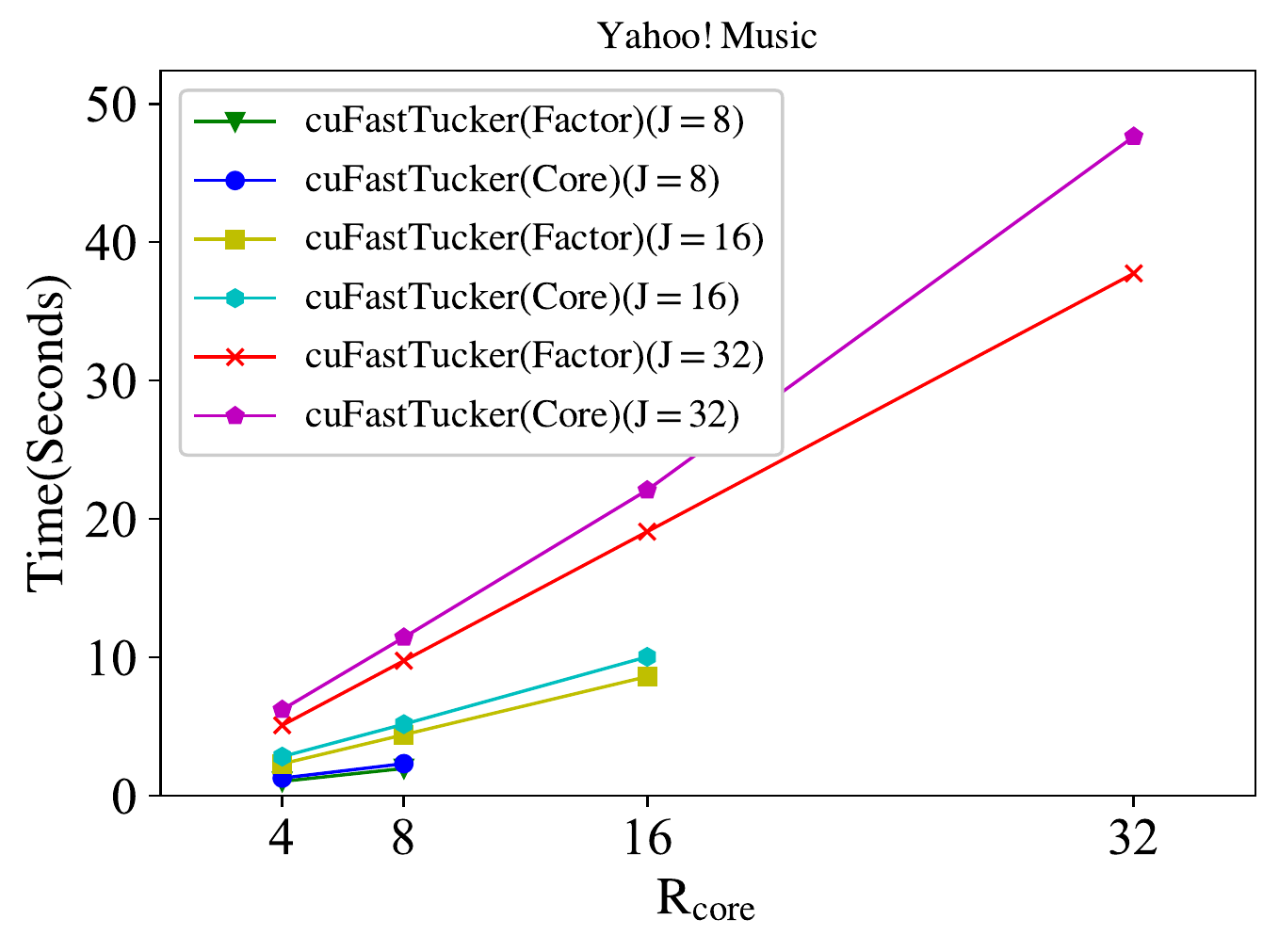}}
    \caption{
     Time (Seconds) on Netflix (a) and Yahoo!Music (b) in various parameters $J_{n}=\{2^2, 2^3, 2^4, 2^5\}$, $n\in\{N\}$, and
     Time (Seconds) on Netflix (c) and Yahoo!Music (d) in various parameters $R_{core}=\{4, 8, 16, 32\}$.}
    \label{Comparison_of_others_approaches_time_dimen}
\end{figure*}
\section{Experiments} \label{section5}
This section mainly answers the following main questions:
(1) the influence of the parameters $\{J_{n}, R_{core}|n\in\{N\}\}$, and the computational overhead of each part of cuFastTucker (Section \ref{section51});
(2) the accuracy performance of cuFastTucker (Section \ref{section52});
(3) the scalability of cuFastTucker and the performance of cuFastTucker on multi-GPUs (Section \ref{section53}).
In this section, cuTucker, P-tucker\cite{8509325}, Vest\cite{2021VEST} and SGD\_Tucker\cite{2020SGD} are taken as the comparison methods.
cuTucker is denoted as the stochastic strategy for STD without Kruskal product for approximating the core tensor on GPU CUDA programming.
SGD\_Tucker\cite{2020SGD} is denoted as the stochastic strategy for STD without the reduction strategy of the computational overhead presented in Theorems \ref{theorem1} and \ref{theorem2}.
Vest\cite{2021VEST} is the parallel CCD method of STD and P-Tucker\cite{8509325} is the parallel ALS method of STD.

\subsection{Experimental Setup} \label{section51}
The experiments are ran on \textbf{Intel(R) Xeon(R) Silver 4110 CPU @ 2.10GHz} with 32 processors and 4 \textbf{NVIDIA Tesla P100 GPUs} with CUDA version 10.0.
The experimental datasets are divided into real and synthesis sets in Tables \ref{data_sets_real} and \ref{data_sets_simulation}, respectively.
The 3 real world datasets are listed as:
Netflix\footnotemark[1] \footnotetext[1]{https://www.netflixprize.com/},
Yahoo!Music\footnotemark[2] \footnotetext[2]{https://webscope.sandbox.yahoo.com/}
and Amazon Reviews \footnotemark[3] \footnotetext[3]{http://frostt.io/tensors/amazon-reviews/}.
The Netflix and Yahoo!Music datasets are used to get the baseline accuracy, and the Amazon Reviews dataset is used to test the ability of cuFastTucker on large-scale data.
The 8 synthesis datasets are produced to test the overall performance of cuFastTucker.
The accuracy is measured by $RMSE$ as
$\sqrt{\bigg(\sum_{(i,j)\in\Gamma}(v_{i,j}-\widetilde{v}_{i,j})^{2}\bigg)\bigg/|\Gamma|}$,
and $MAE$ as $\sum_{(i,j)\in\Gamma}|v_{i,j}-\widetilde{v}_{i,j}|\bigg/|\Gamma|$,
where $\Gamma$ is the test dataset.

The dynamic learning rate of cuTucker and cuFastTucker uses the strategy in \cite{ex165} as $\gamma_t=\frac{\alpha}{1+\beta \cdot t^{1.5}}$,
where the parameters $\{\alpha, \beta, t, \gamma_t\}$ represent the initial learning rate, adjusting parameter of the learning rate, the number of current iterations, and the learning rate at $t$ iterations, respectively.
The parameters in cuTucker and cuFastTucker are listed in Table \ref{the_parameters_of_cuTucker} and Table \ref{the_parameters_of_cuFastTucker}, respectively.
$\{\alpha_a, \beta_a, \lambda_a\}$ are denoted as the parameters for updating the feature matrix in cuTucker and cuFastTucker, and
$\{\alpha_g, \beta_g, \lambda_g\}$ and $\{\alpha_b, \beta_b, \lambda_b\}$ are denoted as the parameters for updating the core tensor in cuTucker and cuFastTucker, respectively.

\subsection{Influences of Various Parameters} \label{section52}
Except for the learning parameters $\{\alpha, \beta, t, \gamma_t\}$,
the value of $\{J_n, R_{core}| n\in\{N\}\}$ determine the training time overhead, space overhead and overall accuracy.
Fig. \ref{Impact_of_different_parameters} illustrates the accuracy influence of the baseline algorithm (cuTucker) and cuFastTucker in various parameters $R_{core}=\{8, 16, 32\}$ with fixed $J_{n}=\{8, 16, 32\}$, $n\in\{N\}$.
Fig. \ref{Impact_of_different_parameters_core} shows the accuracy influence of the baseline algorithm (cuTucker) and cuFastTucker in various parameters $J_{n}=R_{core}=\{8, 16, 32\}$, $n\in\{N\}$.
Fig. \ref{Comparison_of_others_approaches_time_dimen} depicts the training time overhead according to the varying of the value of $J_{n}$, $n\in\{N\}$ and $R_{core}$.
The influence of access time for shared memory and global memory on GPU is presented in Tables \ref{cuTucker_time}-\ref{cuTucker_time_b}.
The conclusion is listed as the following 2 parts:

(1) As the Figs. \ref{Impact_of_different_parameters} and \ref{Impact_of_different_parameters_core} show, both increasing the value of $R_{core}$ and $J_n$ can increase the accuracy (decrease the value of RMSE and MAE) and the value of $J_n, n\in\{N\}$ plays more influence on bigger dataset (Yahoo!Music) than smaller one (Netflix).
Fig. \ref{Impact_of_different_parameters} shows that when $R_{core}=J_n, n\in\{N\}$,
the accuracy performance (RMSE and MAE) of cuFastTucker will overwhelm the cuTucker,
which means that the core tensor has low-rank inherence and the compression rate is $(\sum_{n=1}^{N}R_{core}J_{n})/(\prod_{n=1}^{N}J_{n})$.
Fig. \ref{Impact_of_different_parameters_core} also illustrates the accuracy performance (RMSE and MAE) of updating factor matrix with core tensor (Factor+Core) and factor matrix only (Factor).
In bigger dataset (Yahoo!Music) the accuracy gap between the curves 'Factor+Core' and 'Factor' is much small than smaller volume one (Netflix).

(2) Figs. \ref{fig503 (a1)} and \ref{fig503 (a2)} compare the time overhead of cuTucker and cuFastTucker.
The time overhead comprises of the overhead of updating factor matrix (Factor) and core tensor (Core).
As the Figs. \ref{fig503 (a1)} and \ref{fig503 (a2)} show, the time overhead of both cuTucker (Factor) and cuTucker (Core) is much higher than cuFastTucker (Factor) and cuFastTucker (Core).
The reason is that the Kruskal approximation and overhead reduction by Theorems \ref{theorem1} and \ref{theorem2} can reduce the computational overhead of cuFastTucker.
As the Figs. \ref{fig503 (a1)} and \ref{fig503 (a2)}, the computational overhead of cuFastTucker is increased linearly with the increasing of the value of $R_{core}$ and $J_{n}$.
Shared memory and global memory are the main memory classes in GPUs.

Due to Kruskal approximation for core tensor of cuFastTucker, the approximation matrix $\textbf{B}^{(n)}, n\in\{N\}$ rather than the core tensor $\mathcal{G}$ and unfolding matrices $\textbf{G}^{(n)}, n\in\{N\}$ can be accessed on shared memory.
As the Tables \ref{cuTucker_time}-\ref{cuTucker_time_b_titan} show,
memory accessing speed on shared memory is slightly faster than global memory.
cuTucker has a huge intermediate matrix when updating the core tensor, 
which cannot be stored in shared memory when $J_n, n\in\{N\}$ is greater than 8. 
While cuFastTucker can store larger core tenso in shared memory, 
this situation is more obvious when the order is larger.
Further, cuFastTucker makes it easier to put memory hotspot core tensor into shared memory,
which will further reduce the computation time.
Since the shared memory and the L1 cache share a piece of on-chip memory, 
the increase of the shared memory will lead to a decrease in the L1 cache, resulting in a slight decrease in program performance.
This is more noticeable on the \textbf{NVIDIA TITAN RTX GPU} with larger caches.

\begin{figure*}[htbp]
  \centering
    \subfigure[RMSE on Netflix]{
    \label{fig505 (a1)} 
    \includegraphics[width=2.8in]{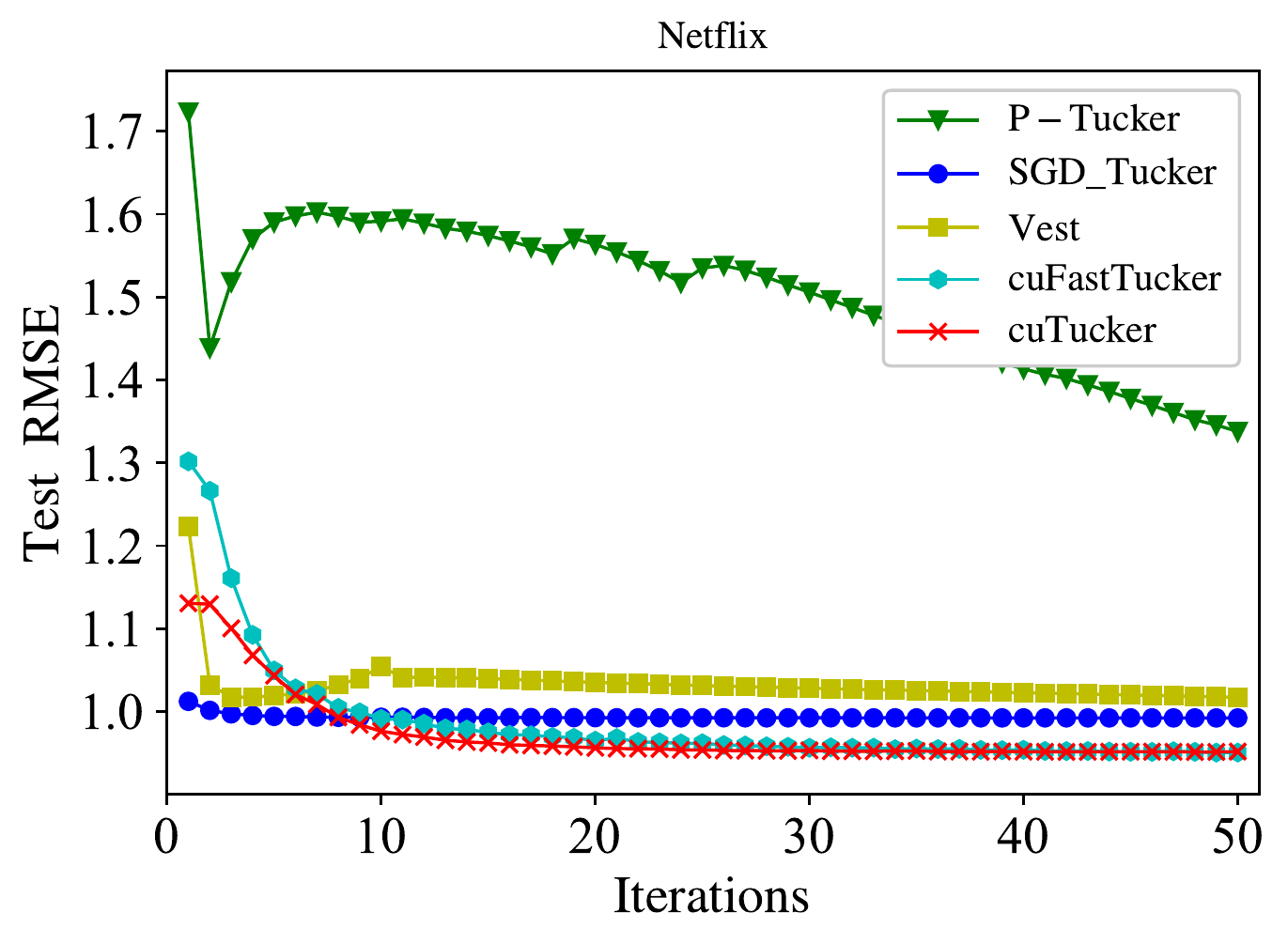}}
    ~
    \subfigure[MAE on Netflix]{
    \label{fig505 (a2)} 
    \includegraphics[width=2.8in]{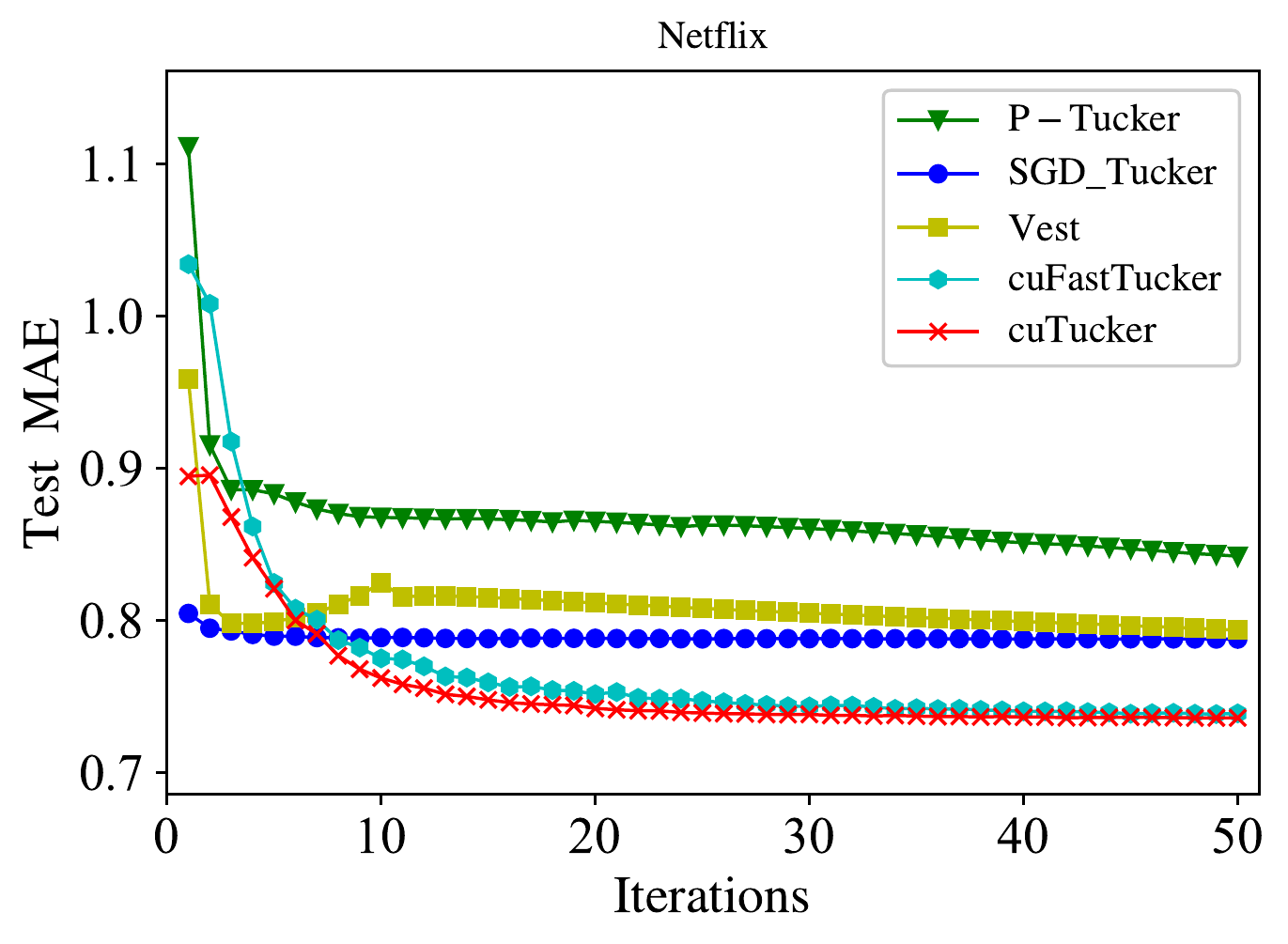}}
    ~
    \subfigure[RMSE on Yahoo!Music]{
    \label{fig505 (a3)} 
    \includegraphics[width=2.8in]{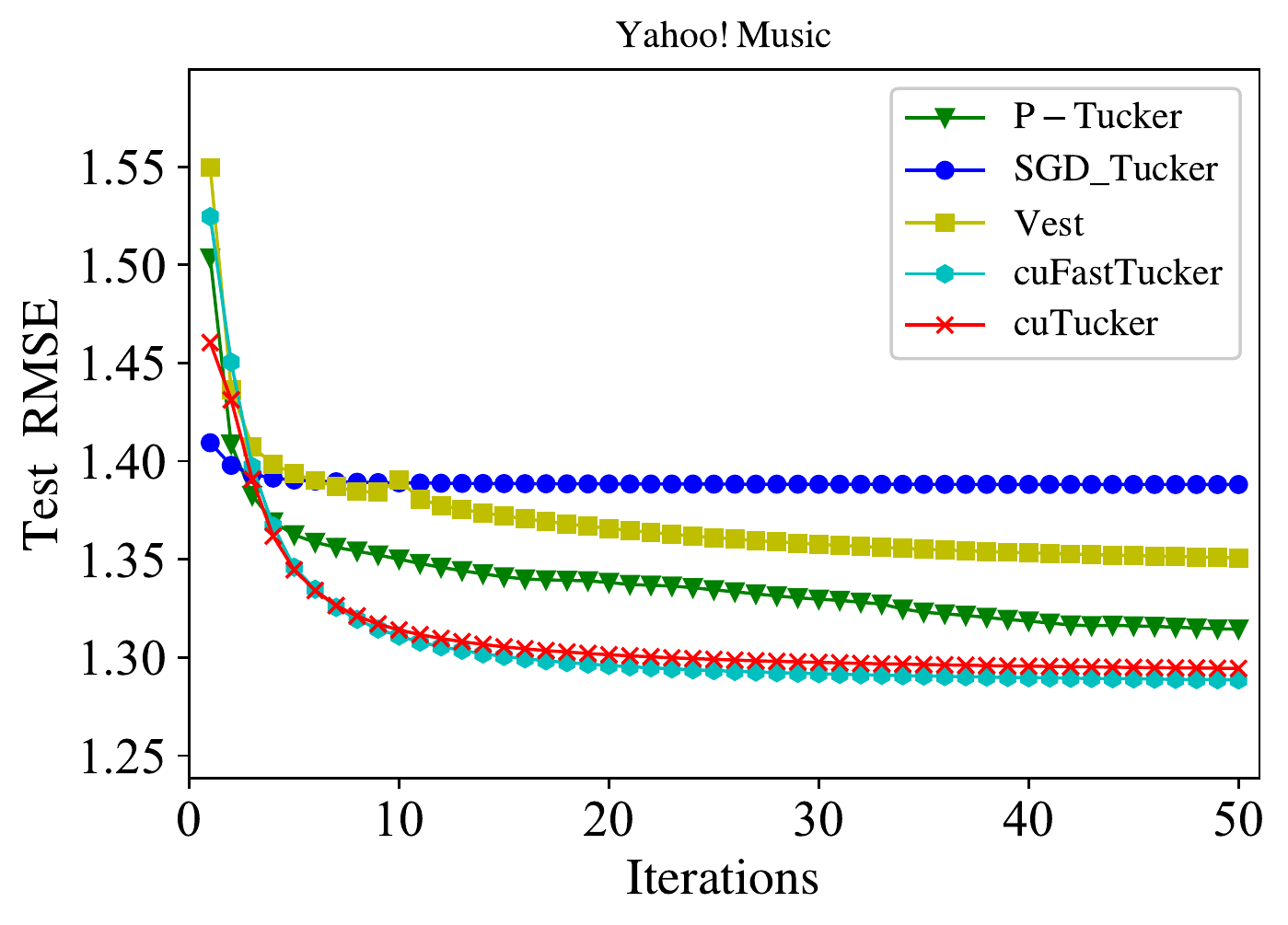}}
        ~
      \subfigure[MAE on Yahoo!Music]{
      \label{fig505 (a4)} 
    \includegraphics[width=2.8in]{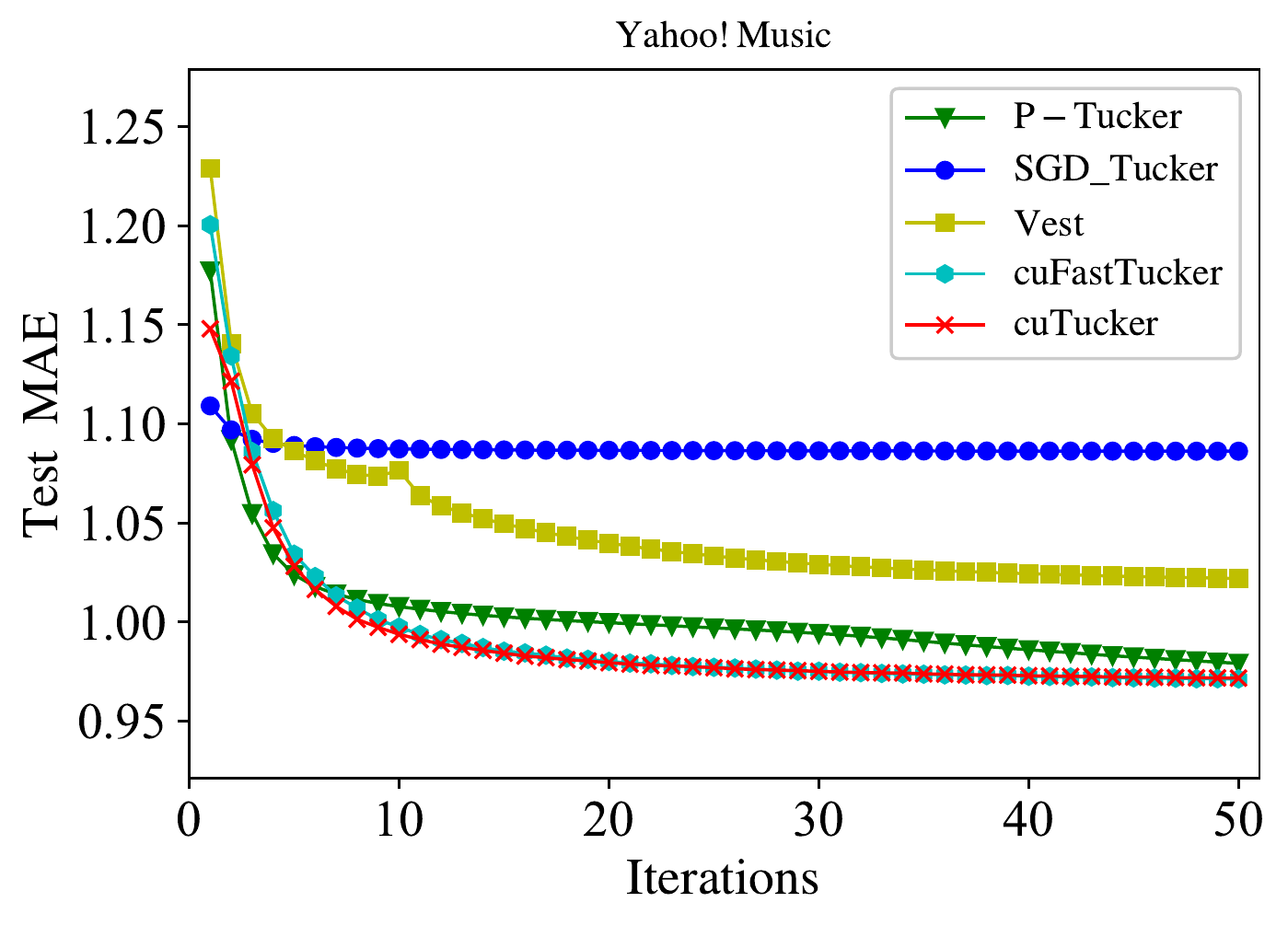}}
         \caption{Comparison with STOA approaches.}
    \label{Comparison_of_others_approaches}
\end{figure*}

\begin{table}[htbp]
	\centering
	\scriptsize
	\setlength{\abovecaptionskip}{0pt}
	\caption{The time overhead (Seconds) to update the low rank factor matrices in a single iteration of STOA algorithms.}
	\begin{tabular}{c|cc}
		\hline
		\hline
						& Netflix        & Yahoo!Music   \\
		\hline
		P-Tucker        & 20.539240 (106.73X)    & 132.954739 (197.57X)   \\
		Vest            & 75.574325 (392.74X)   & 503.045186 (747.54X)  \\
		SGD\_Tucker     & 12.111856 (62.94X)     & 29.152721 (43.32X)    \\
		cuTucker        & 0.696915 (3.62X)       & 1.761734 (2.61X)     \\
		\textbf{cuFastTucker}    & \textbf{0.192428}       & \textbf{0.672929}      \\
		\hline
		\hline
	\end{tabular}
	\label{time_others_approaches}
\end{table}

\subsection{Comparison with STOA Approaches} \label{section53}
Fig. \ref{Comparison_of_others_approaches} and Table \ref{time_others_approaches} depict the comparison of the convergence and accuracy performances and
the running time per iteration is presented on Table \ref{time_others_approaches}.
P-tucker\cite{8509325}, Vest\cite{2021VEST} and SGD\_Tucker\cite{2020SGD} are CPU based method.
Meanwhile, cuTucker and cuFastTucker are GPUs based methodologies.
To ensure the running fairness,
the CPU and GPU run independently without the interference of other works.
All comparison methodologies run on $J_{n}=4,n \in \{N\}$, and in cuFastTucker runs on $R_{core}=4$.
Some algorithms lack the update of the core tensor, and
we only compare the update of the factor matrix here.
As the Fig. \ref{Comparison_of_others_approaches} show, P-tucker has the fastest RMSE decreasing speed at the beginning but it goes slower in the later stage.
P-Tucker runs a bit of unstable.
All the methods can obtain the same overall accuracy after 20 iterations.
SGD\_Tucker runs much faster than P-Tucker and Vest but a bit slower than cuTucker and cuTastTucker.
The convergence speed and accuracy of cuTucker and cuFastTucker overwhelm the other three algorithms.
As show in Table \ref{time_others_approaches}, cuFastTucker and cuTucker get the top-2 and due to Kruskal approximation and
the overhead reduction of vectors multiplication in Theorems \ref{theorem1} and \ref{theorem2}, cuFastTucker obtain 3.62X and 2.61X speedup than cuTucker on Netflix and Yahoo!Music datasets, respectively.

\begin{figure*}[htbp]
  \centering
    \subfigure[Scalability on Synthesis Datasets]{
  	\label{fig508 (a1)}
	\includegraphics[width=2.0in]{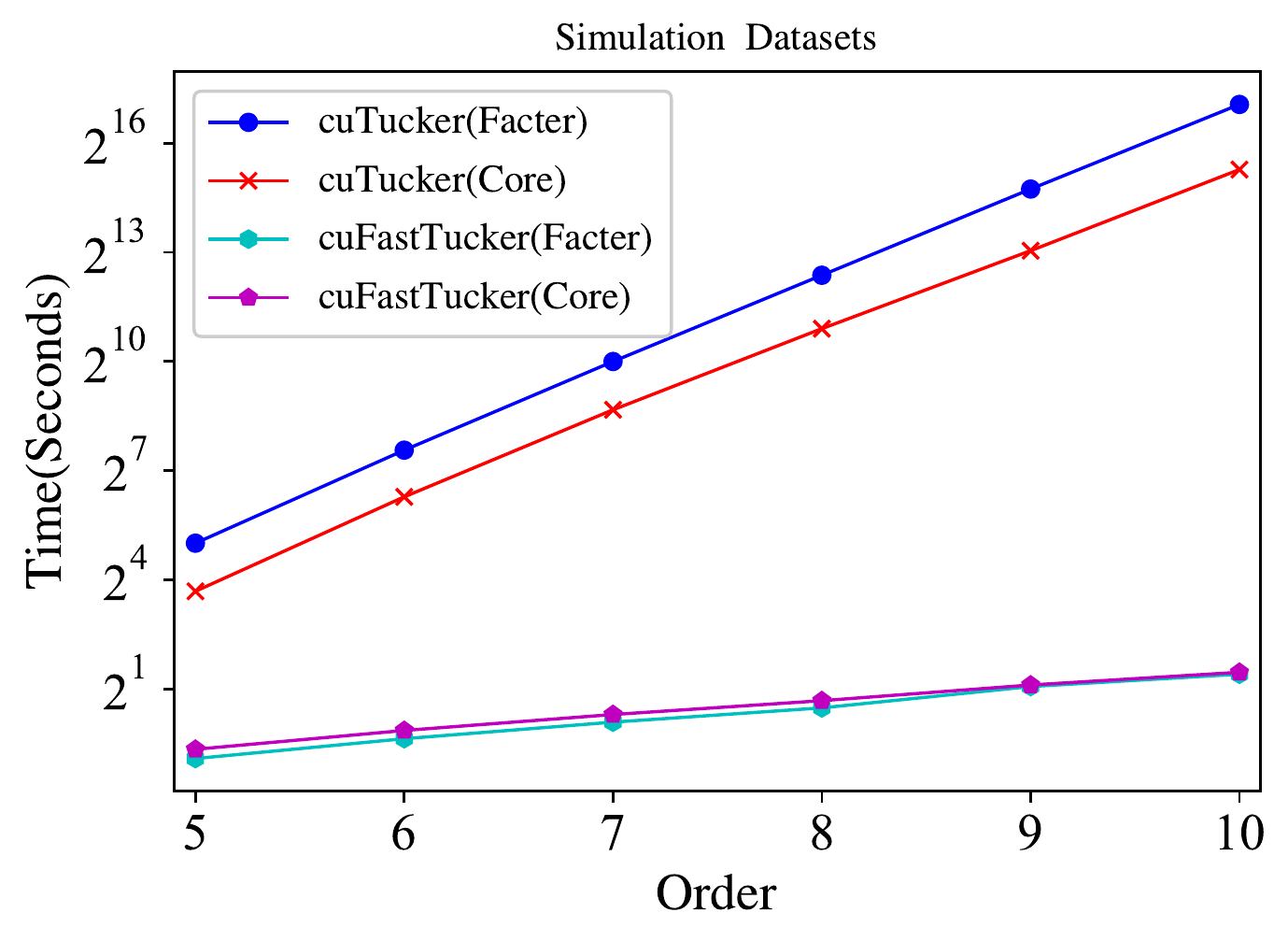}}
    ~
    \subfigure[Speedup on Netflix]{
    \label{fig508 (a2)} 
    \includegraphics[width=2.0in]{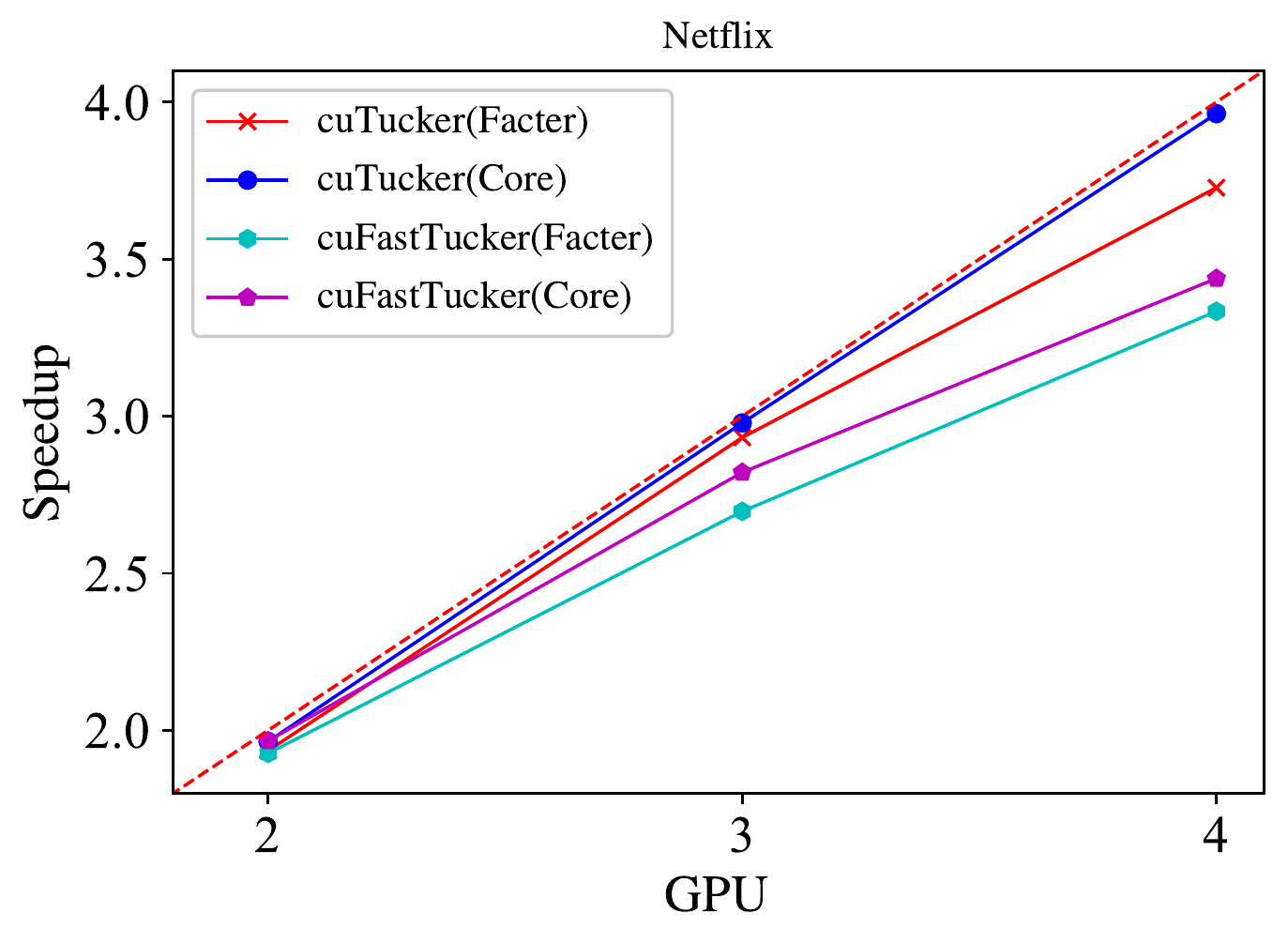}}
    ~
    \subfigure[Speedup on Yahoo!Music]{
    \label{fig508 (a3)} 
    \includegraphics[width=2.0in]{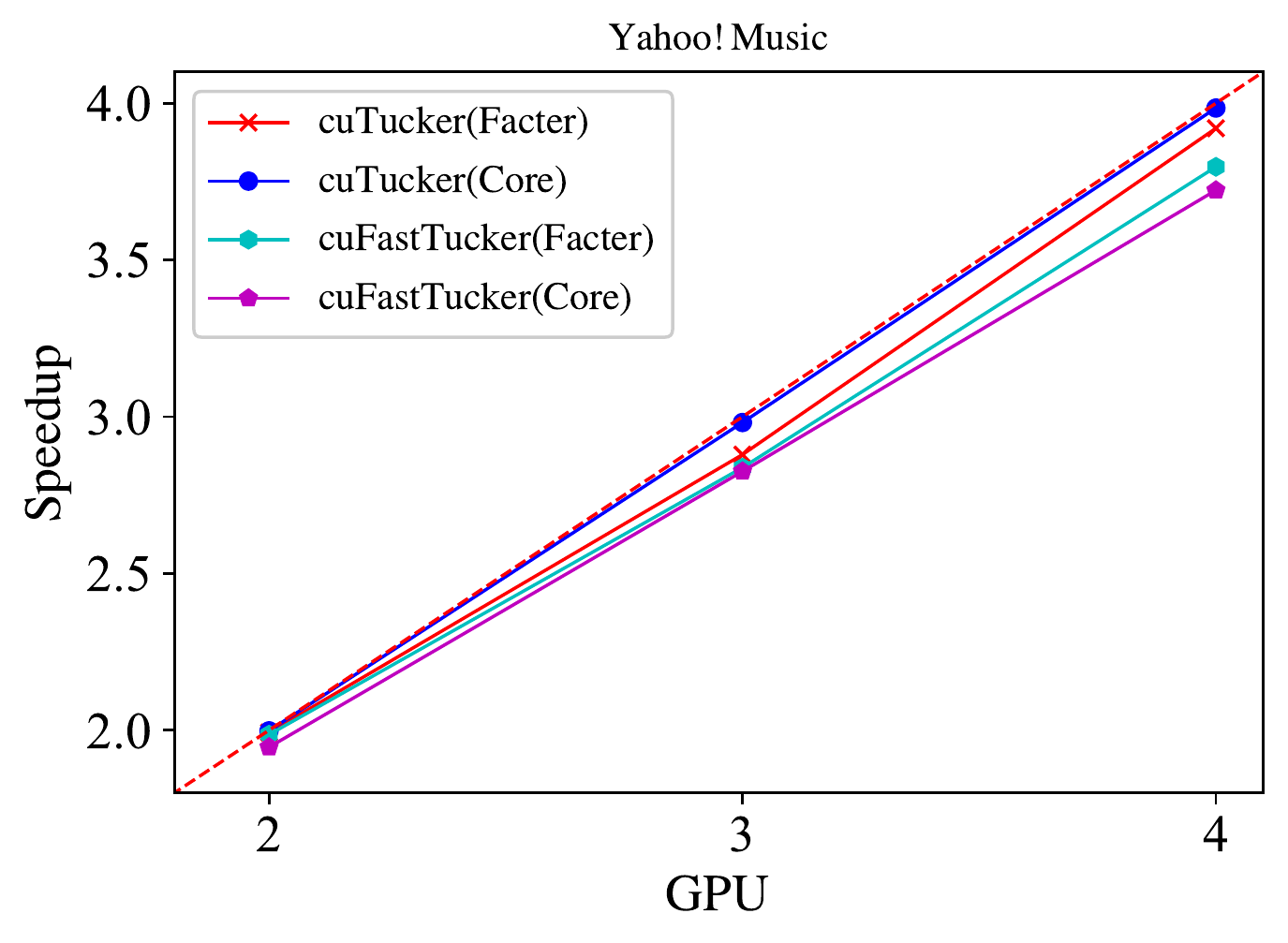}}
         \caption{(a) Scalability on synthesis datasets; (b)-(c) speedup on Netflix and Yahoo!Music datasets, respectively.}
    \label{Scale_up_to_multiple_GPUs}
\end{figure*}

\begin{figure*}[htbp]
  \centering
    \subfigure[Order=3]{
    \label{fig509 (a1)} 
    \includegraphics[width=2.00in]{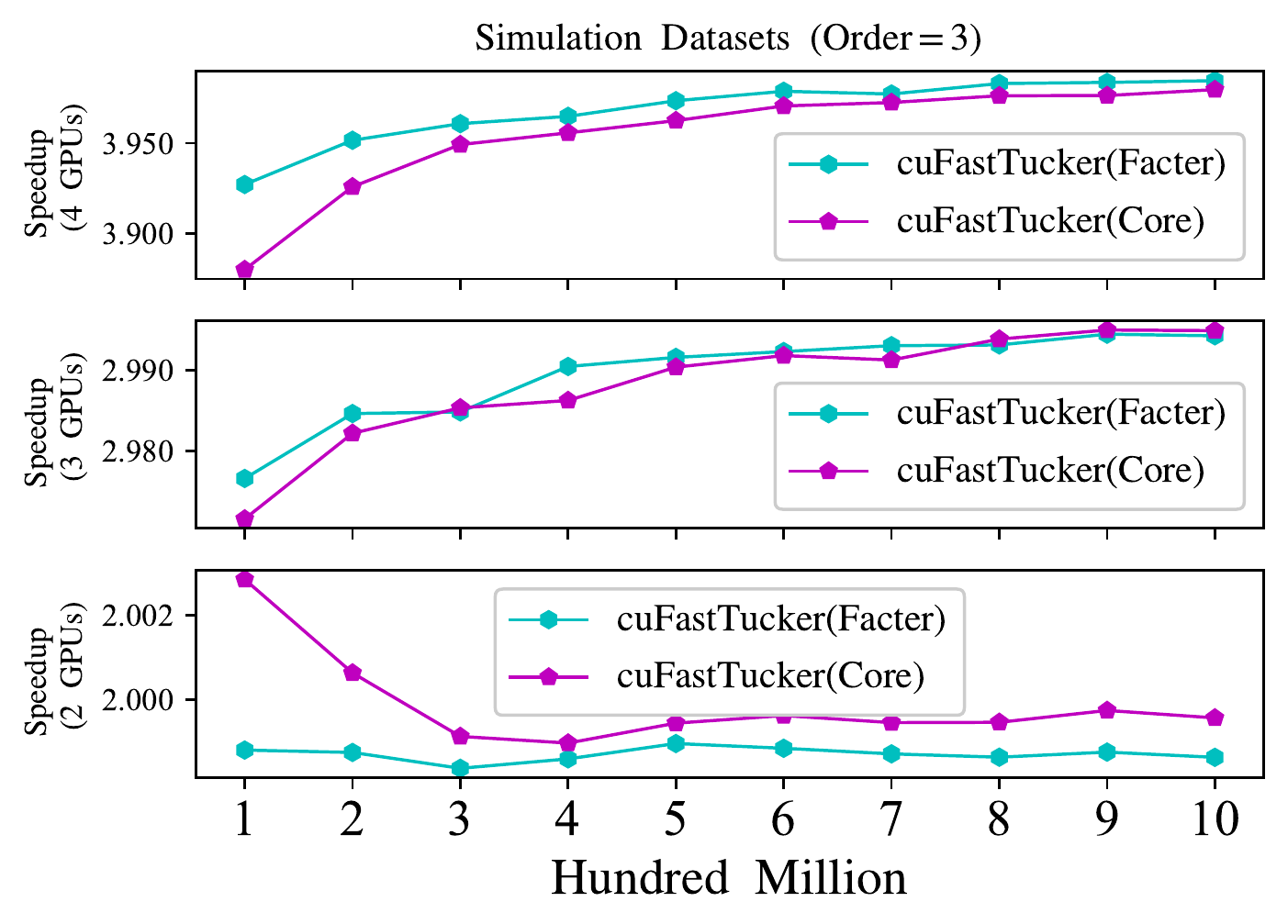}}
    ~
    \subfigure[Order=4]{
    \label{fig509 (a2)} 
    \includegraphics[width=2.00in]{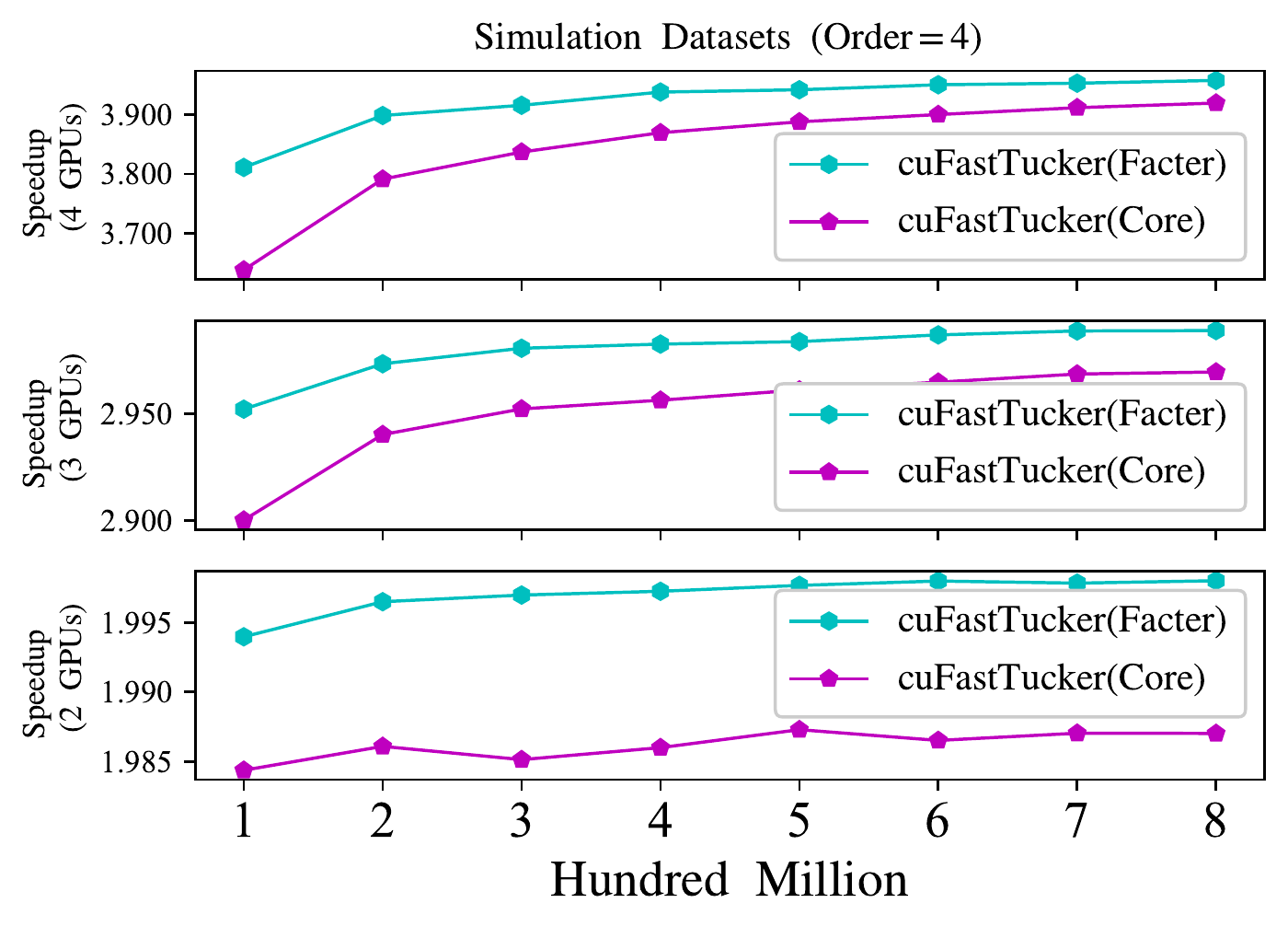}}
    ~
    \subfigure[Order=5]{
    \label{fig509 (a3)} 
    \includegraphics[width=2.00in]{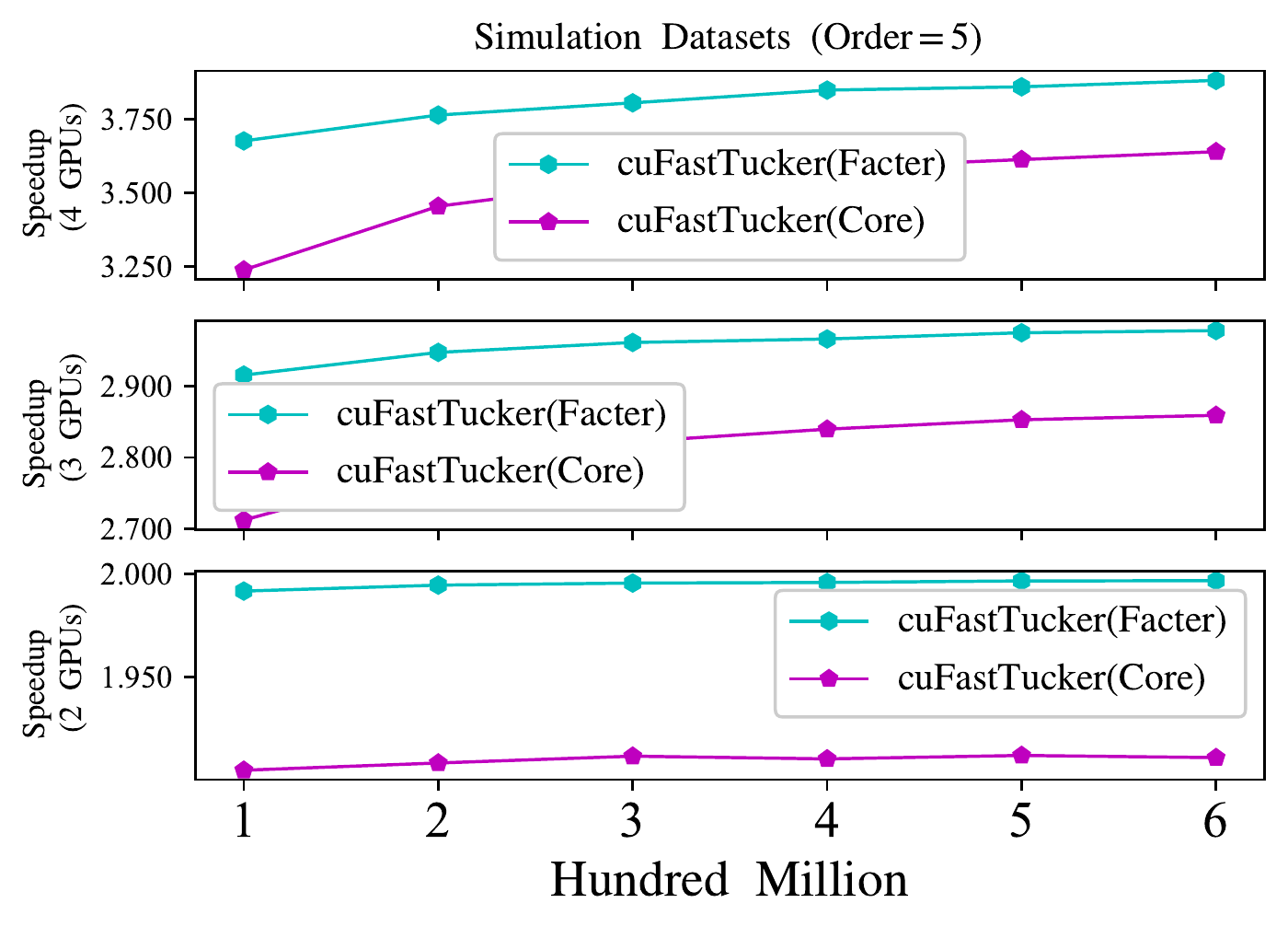}}
         \caption{Scale up to multiple GPUs}
    \label{Scale_up_to_multiple_GPUs_Simulation}
\end{figure*}

\subsection{Scalability on Large-scale Data and Speedup on Multi-GPUs} \label{section54}
Fig. \ref{fig508 (a1)} illustrates the computational time of updating core tensor and factor matrix of cuTucker and cuFastTucker.
Both the cuTucker and cuFastTucker have the scalability with the order of $\{5, 6, 7, 8, 9, 10\}$.
Whatever updating the core tensor and factor matrix, cuTucker spends much longer time than cuFastTucker.
Figs. \ref{fig508 (a2)} and \ref{fig508 (a3)} present the speedup on $\{2, 4, 5\}$ GPUs and both the cuTucker and cuFastTucker can obtain the near linear speedup.
Fig. \ref{Scale_up_to_multiple_GPUs_Simulation} shows that with fixed order, cuFastTucker can obtain more stable speedup on the high non-zero entries datasets.
From the results of Figs. \ref{fig508 (a2)} and \ref{fig508 (a3)} and \ref{fig509 (a1)}-\ref{fig509 (a3)}, cuFastTucker can get near the linear speedup on synthesis datasets.
For the very large dataset amazn, cuFastTucker runs perfectly on 4 P100s. When $R_{core}=J_n=4, n\in\{N\}$, the time for a single update of the factor matrix and core tensor are 10.769747 seconds and 12.953006 seconds, respectively.

\section{Conclusion}
\underline{H}igh-\underline{O}rder, \underline{H}igh-\underline{D}imension, and \underline{S}parse \underline{T}ensor (HOHDST) is a widely used data form in ML community, etc,
spatiotemporal dynamics social networks and recommender systems, and network flow prediction.
Thus, it is non-trivial to find an efficient and low computational overhead methodologies for Sparse Tensor Decomposition (STD) to get the key feature of HOHDST data.
To solve this problem, cuFastTucker is proposed which comprise of Kruskal core tensor and Theorems \ref{theorem1} and \ref{theorem2} to reduce the computational overhead.
Meanwhile, low data-dependence gives the cuFastTucker with fine-grained parallelization on CUDA GPU.
The experimental results show that cuFastTucker has linear computational time and space overheads and
cuFastTucker runs at least 2.62$X$ and at most 747.54$X$ faster than STODA approaches.
In the future works,
we will explore how to take advantages of cuFastTucker to accelerate and compress modern Deep Neural Networks, etc, CNN, LSTM, RNN, and Transformer.

\ifCLASSOPTIONcompsoc
  \section*{Acknowledgments}
\else
  \section*{Acknowledgment}
\fi
This work has also been partly funded by the Program of National Natural Science Foundation of China (Grant No. XXXXXXXXXX),
the National Outstanding Youth Science Program of National Natural Science Foundation of China (Grant No. XXXXXXXXXX).

\ifCLASSOPTIONcaptionsoff
  \newpage
\fi

\bibliographystyle{IEEEtran}
\bibliography{bib}

\end{document}